\documentclass[10pt]{article}
\usepackage{amsfonts}
\usepackage{amssymb}
\usepackage{amsthm}
\usepackage{amsmath}    
\usepackage{bbm}
\usepackage{graphicx}
\usepackage{latexsym}
\usepackage{mathrsfs}
\usepackage[all]{xy}

\usepackage{geometry}
\newcommand{\cA}{\mathcal{A}}
\newcommand{\cB}{\mathcal{B}}

\newcommand{\cD}{\mathcal{D}}

\newcommand{\cF}{\mathcal{F}}
\newcommand{\cG}{\mathcal{G}}
\newcommand{\cH}{\mathcal{H}}

\newcommand{\cM}{\mathcal{M}}

\newcommand{\cP}{\mathcal{P}}
\newcommand{\cQ}{\mathcal{Q}}

\newcommand{\cS}{\mathcal{S}}

\newcommand{\cU}{\mathcal{U}}

\newcommand{\rC}{\mathrm{C}}
\newcommand{\rD}{\mathrm{D}}

\newcommand{\rF}{\mathrm{F}}

\newcommand{\rH}{\mathrm{H}}

\newcommand{\rK}{\mathrm{K}}
\newcommand{\rL}{\mathrm{L}}

\newcommand{\rN}{\mathrm{N}}

\newcommand{\rT}{\mathrm{T}}

\newcommand{\rd}{\mathrm{d}}

\newcommand{\rg}{\mathrm{g}}
\newcommand{\rh}{\mathrm{h}}

\newcommand{\rr}{\mathrm{r}}

\newcommand{\rt}{\mathrm{t}}
\newcommand{\ru}{\mathrm{u}}

\newcommand{\sA}{\mathscr{A}}

\newcommand{\sP}{\mathscr{P}}

\newcommand{\bC}{\mathbb{C}}

\newcommand{\bM}{\mathbb{M}}
\newcommand{\bN}{\mathbb{N}}

\newcommand{\bZ}{\mathbb{Z}}

\newcommand{\Si}{\Sigma} 

\newcommand{\si}{\sigma}

\newcommand{\io}{\iota}

\title{\textsc{Causal posets, loops and the construction of nets of local algebras for QFT}} 

\author{\textsc{Fabio Ciolli}$^{a}$ \ \   \textsc{Giuseppe Ruzzi}$^{a}$ \ \    \textsc{Ezio Vasselli}$^{b}$\footnote{All the authors are supported in part by the  EU network ``Noncommutative Geometry" MRTN-CT-2006-0031962. 
F.C. is supported by the ERC Advanced Grant 227458 OACFT ``Operator Algebras and Conformal Field Theory".}\\[5pt]
\small{$^{a}$  Dipartimento di Matematica, Universit\`a di Roma ``Tor Vergata'',}\\
\small{Via della Ricerca Scientifica 1, I-00133 Roma,  Italy.}  \\
\small{\texttt{ciolli@mat.uniroma2.it}, \texttt{ruzzi@mat.uniroma2.it}} \\[5pt]
\small{$^{b}$Dipartimento di Matematica, Universit\`a di Roma ``La Sapienza'',}\\
\small{Piazzale Aldo Moro 5, I-00185 Roma, Italy.} \\
\small{\texttt{ezio.vasselli@gmail.com}}}
\date{}

\begin{document}

\maketitle

\begin{center}
\emph{Dedicated to the memory of Claudio D'Antoni}
\end{center}

\vspace{12pt}

\begin{abstract}
We provide a model independent construction of a net 
of $\rC^*$-algebras satisfying the Haag-Kastler axioms 
over any spacetime manifold. Such a net, called \emph{the net of causal loops},  
is constructed by selecting a suitable base $K$ 
encoding causal and symmetry properties of the spacetime. 
Considering $K$ as a partially ordered set 
(poset) with respect to the inclusion order relation,  we  define  groups of closed paths (loops)
formed by the elements of $K$. These groups come equipped  with a  causal disjointness relation and  an action of the symmetry group of the spacetime. In this way the local algebras 
of the net are  the group $\rC^*$-algebras of the groups of loops, quotiented  by the causal disjointness relation. We also provide a geometric interpretation of a class of representations of this  net   in terms of  causal and covariant connections of the poset $K$.  
In the case of the Minkowski spacetime, we prove the existence of Poincar\'e covariant representations  satisfying the spectrum condition. 
This is obtained by virtue of a remarkable feature of our construction: any Hermitian scalar quantum field defines causal and covariant connections of $K$.
Similar results hold for the chiral spacetime $S^1$ with conformal symmetry.
\end{abstract}

\newpage 

\tableofcontents
\markboth{Contents}{Contents}


  \theoremstyle{plain}
  \newtheorem{definition}{Definition}[section]
  \newtheorem{theorem}[definition]{Theorem}
  \newtheorem{proposition}[definition]{Proposition}
  \newtheorem{corollary}[definition]{Corollary}
  \newtheorem{lemma}[definition]{Lemma}

  \theoremstyle{definition}
  \newtheorem{remark}[definition]{Remark}
    \newtheorem{example}[definition]{Example}

\theoremstyle{definition}
  \newtheorem{ass}{\underline{\textit{Assumption}}}[section]


\numberwithin{equation}{section}

\newpage

\section{Introduction}
\label{I}

Quantum field theory is nowadays a  framework  for  understanding  the physics of relativistic quantum systems, notably 
elementary particles,  and the effects of gravitation on quantum systems, at least in the approximation in which  the gravitational  
field can be treated as a background field.  Several mathematical approaches are appropriate for this aim: 
operator valued distributions; probability measures on the space of distributions;  
causal and covariant nets of operator algebras. 
All these approaches, distinguished from others because non-perturbative, 
are strictly related and result to be,  in the Minkowski spacetime  and under suitable assumptions, equivalent 
(see \cite{GJ} and \cite{Buc,Bos}).\\   
\indent The mathematically rigorous construction of quantum fields is object of research since the middle fifties. Many important 
ideas have been developed throughout these years although the main aim, i.e.\  the construction of a non-perturbative interacting quantum field in the 4-dimensional  Minkowski spacetime, is still an open problem.  In the present paper we give a novel construction of quantum fields which applies on any 
meaningful spacetime. We shall work in the context of the algebraic quantum field theory (AQFT), 
the last approach  among those  listed above.\\
\indent The main idea of AQFT \cite{Haa} is that the  physical content of a theory is encoded in the \emph{observable net}: 
the correspondence between  a suitable family $K$ of bounded 
regions of the spacetime  and the  $\rC^*$-algebras  generated by  observables measurable within the elements of $K$. 
This correspondence is assumed to preserve the inclusion of regions, to satisfy Einstein causality 
and to be covariant with respect to the symmetry group of the spacetime  (Haag-Kastler axioms). In the Minkowski spacetime, the just described scenario turns out to be appropriate  for understanding, in term of superselection sectors of the observable net, the charge structure and the statistics of particle physics \cite{DHR, BF} as well as the presence of \emph{a global} gauge group acting  upon  a net of unobservable quantum fields \cite{DR}\footnote{Charges of electromagnetic type do not fit in superselection sectors analyzed \cite{DHR, BR,DR}. Progresses in this direction should  appear soon \cite{DopConf}.}.\\ 
\indent It is worth pointing out the relevance of the set of regions $K$ of the spacetime: this is usually called \emph{the set of indices} of the observable net and is took to be a base of the topology of the spacetime itself. It is also required  that it may codifies the  symmetry and the causal structure of the spacetime. Examples of set of indices $K$ are double cones for the Minkowski spacetime; open intervals for the circle $S^1$; diamonds for an  arbitrary globally hyperbolic spacetime. The relevance of $K$, 
as a partially ordered set (\emph{poset}) with respect to the inclusion order relation,  has been strengthened by the formulation of 
the theory of superselection sectors in terms of the non-Abelian cohomology of $K$ (see \cite{Rob} for a review); this also allowed to extend the analysis of sectors to curved spacetimes \cite{GLRV, Ruz} and to introduce a generalized notion of sector associated to 
invariants of the spacetime described by the poset $K$ \cite{BR}.

Concerning the construction of models in the AQFT framework,  apart from the well known approach of the Weyl algebras,   i.e.\
a version of canonical quantization,  it is worth pointing out the technique of modular localization \cite{BGL} 
(based on the ideas in \cite{Sch}) allowing to construct a causal and covariant net of operator algebras  
from a given representation of the Poincar\'e group.
However,  important progresses, especially towards the construction of interacting theories, have recently been achieved 
 by the introduction of a new approach based on   operator algebras   deformation \cite{GL} (see also  \cite{BLS}).   
Finally, it is worth mentioning  the abstract construction of Borchers-Uhlmann \cite{Bor,Uhl}  
where a causal and covariant net of  $^*$-algebras over the Minkowski spacetime is constructed using merely 
test functions,   and the properties of the localization of their supports.   A remarkable feature of this method is
that  quantum fields satisfying the Wightmann axioms  
are in bijective correspondence with a class of positive linear functionals of this net.  
\smallskip 

Our work is motivated by quantum gauge theories, a scenario where meaningful quantities  are obtained  evaluating 
fields over paths.
The aim of the present paper is to construct
a causal and covariant net of $\rC^*$-algebras over  any spacetime,  using paths as the main ingredient.
This is done in a model independent way, regardless of any specific theory of quantum field over the spacetime itself. 
In this way, we give a new construction of quantum field by using as input merely the spacetime
(to be precise, the base $K$) and  its causal and symmetry structure;
at the same time,  we hope to capture some structural properties
of any reasonable formulation of a quantum gauge theory in the framework of algebraic quantum field theory.\\
\indent The notion of path that we use here is that used in the context of abstract posets \cite{Rob,Rob2}. 
Given a  base $K$ of a spacetime ordered under inclusion, a path in $K$ can be figured out  
as  an approximation of a path $\gamma$, intended in the usual topological way, 
by means of a finite sequence of elements of  $K$ covering it. 
%
This notion of path is good enough to capture nontrivial invariants 
of the spacetime, as the fundamental group \cite{Ruz} and the category of flat bundles \cite{RRV};
moreover,  paths are the key for the definition of connection over a poset \cite{RR}, in essence a group valued map defined on the set of paths of $K$, 
designed to give a version of the basic structures of gauge theory in the setting of AQFT.\smallskip

The first step of our construction is made in Sections  \ref{A} and \ref{B}.  
We define the category \textbf{Causet} of posets endowed with a causal disjointness relation (\emph{causal posets})  as
model of bases of spacetimes. 
%
%
Afterwards we construct a {\em causal functor} from \textbf{Causet} to the category of $\rC^*$-algebras\footnote{A causal functor generalizes  the functor underlying the definition of a locally covariant quantum field theory \cite{BFV}. The only difference is that  the source category of the latter is that of 4-d globally hyperbolic spacetimes. 
Since a causal poset can be associated to any spacetime, a locally covariant quantum field theory can be obtained from any  causal functor.}. This is done by assigning to any causal poset the group $\rC^*$-algebra of 
the free group generated by  closed paths (\emph{loops}) quotiented by an induced  causal disjointness relation.   
Using this functor we show that  a causal net of $\rC^*$-algebras can be constructed over any causal poset and,
consequently, over any spacetime. This net, called \emph{the net of causal loops},  turns out to be covariant 
under any symmetry group of the spacetime. We show  the non-triviality of the net of causal loops 
for  the circle $S^1$, for  the  Minkowski spacetime and, more in general, for an arbitrary 
globally hyperbolic spacetime.

%
%

Concerning the representation theory, the  net of causal loops  is  universal with respect to the notion of connection given in  
\cite{RR}: this means that a connection  defined on a Hilbert space which is causal and covariant, in a suitable sense, 
yields a representation of the net. An interesting point is that it seems that
not all the representations of the net arise from connections. Indeed, using a procedure similar to  the reconstruction of the potential from  the free quantum electromagnetic field, we associate to any representation  a \emph{connection system}: 
a system of connections where causality and covariance arise as   properties of the system 
and not  of a single connection, Section \ref{Cc}.  This also leads to a natural notion of gauge transformations.\\  
\indent The passage from a function of loops, the representation, to a function of paths, the connection,  relies on a sort of
reference frame for the poset $K$: the  \emph{path-frame}, a choice of a collection of  paths joining any element of $K$ 
to a fixed one, the pole.  The resulting connection depends on this choice 
and turns out to be  neither causal nor covariant.  However, when the poset admits a covariant family of path-frames, 
then we obtain  the connection system mentioned above. 
In general, this procedure  is subjected to restrictions due to the topology
of the spacetime and to the symmetry group. In particular,  the Minkowski spacetime with
the Poincar\'e symmetry group admits a covariant family of path-frames whilst $S^1$ 
with  the conformal symmetry group does not.

The construction of the net of causal loops is purely combinatorial and no topological condition 
is imposed, even if the symmetry group  is a topological group.  An important result  in this regards  is that, in the Minkowski spacetime,  the net of causal loops  admits Poincar\'e covariant representations fulfilling
the spectrum condition, Section  \ref{Ce}. This is obtained by a remarkable feature of our construction:  any Hermitian  scalar  quantum field yields a causal and Poincar\'e covariant connection which results  to be manifestly nontrivial in the case of the free scalar field. 
Similar results hold for the net of causal loops over the chiral spacetime $S^1$ with conformal symmetry.


\section{Causal posets, functors and nets of $\rC^*$-algebras} 
\label{A}

The order relation structure enters the theory of quantum fields in the correspondence $o\mapsto\cA_o$, 
known as \emph{the observable net},  
associating to the bounded region $o$ of the spacetime the algebra $\cA_o$ generated 
by the observables measurable within $o$. This correspondence is clearly \emph{isotonous}  with  respect to the  
inclusion of regions of the spacetime.  On the other hand, for the applications to quantum field theory,  it is enough to restrict this correspondence to a family  $K$  of regions, usually called \emph{the set of indices}, encoding  topological,  causal and  symmetry 
properties of the underlying spacetime\footnote{This correspondence underlies the  quantum field theory and is the corner stone 
of AQFT, also known as  local quantum physics 
\cite{ Ara, BrF, BH, GJ, Haa, Rob, SW}.}.
These aspects  lead to the notion of a causal poset with a symmetry: the set of indices $K$
has to be considered as a  a partially ordered set (poset) because of the isotony 
property of the observable net; symmetry and causality of the spacetime lift to corresponding properties of the poset $K$.  
From this point of view, the observable net is nothing but  a functor from the poset $K$, considered as a category, 
to the category of $\rC^*$-algebras. This functor preserves, in a suitable sense, the causal and  the symmetry property 
of $K$.  \smallskip  

In this section, taking cue from the locally covariant quantum field theory \cite{BFV}, we generalize the above ideas  
and introduce the notion of a causal functor: a functor from the category formed by causal posets to the one of $\rC^*$-algebras. 
We show that if a causal functor is given,  then  a causal net of $\rC^*$-algebras is associated to any spacetime. 
The existence and the construction of a causal functor shall be given in Section \ref{B}. 
%
%

%
%

\subsection{Causal posets}
\label{Aa}

A poset is a nonempty  set $K$ endowed with a (binary, reflexive, transitive and antisymmetric) 
order relation $\leq$.  We say that $K$ is {\em upward directed} whenever for any $a,a' \in K$ there is
$o \in K$
such that $a , a' \leq o$.  
A poset $K$ is \emph{pathwise connected} if for any pair $a,a'\in K$ 
there are two finite sequences $a_1,\ldots a_{n+1}$ and  $o_1,\ldots o_n$ of elements of $K$,  
with $a_1=a$ and $a_{n+1}= a'$, 
satisfying the relations 
\begin{equation}
\label{Aa:1}
   a_i,a_{i+1}\leq o_{i} \ , \ i=1,\ldots, n \ .  
\end{equation}
Note that an upward directed poset  is pathwise connected. \\
\indent A {\em causal disjointness relation} on a poset $K$ is an irreflexive, symmetric binary relation $\perp$ on $K$
stable under the order relation, given $o,a,\tilde o \in K$ 
\begin{equation}
\label{Aa:2}
o \perp a \ \  , \ \ \widetilde o \leq o 
  \ \ \Rightarrow \ \ 
\widetilde o \perp a \ .
\end{equation}
Note that irreflexivity ($\nexists a\in K$ such that $a\perp a$)  and the stability under the order relation  imply that  two causally disjoint 
elements $a\perp o$ do not have a common minorant, i.e.\,  
$\nexists$ $x\in K$ such that  $x\leq o,a$. \\
\indent We call any  pathwise connected poset 
equipped with a causal disjointness relation 
a \textbf{causal poset}.

\begin{remark}
Two observations are in order.\smallskip

\noindent 1.  The term {\em causal disjointness relation} is used because usually we deal with a poset 
arising as a family, ordered under inclusion, of subsets of a globally hyperbolic spacetime. 
Two events of such spacetimes are {\em causally disjoint} if they cannot be joined by any non-degenerate causal curve and  
this relation extends naturally to the subsets of the spacetime.  Notice, however, that  in particular cases 
causal disjointness equals the set-disjointness. This happens for instance when one considers subsets laying on a  
Cauchy surface.  Furthermore, this is also the case of the circle $S^1$, the spacetime of chiral theories,  
that is obtained as the compactification of one of the lightline of the 2-d  Minkowski spacetime. \smallskip 

\noindent 2.  It is worth stressing the difference between the notion of a \emph{causal poset} and that of a \emph{causal set}. 
The causal set approach to quantum gravity, see \cite{Sor, Sur}, is based on the idea that the spacetime 
is a discrete set $C$ of events $e$ related by an  order relation $\preceq$ describing the causal relation between 
the events of the spacetime (discrete means that the poset $(C,\preceq)$ is locally finite).  Hence,  the notion of causal set is more primitive than that of a causal poset.  In fact,  let $K$ be  the set of the finite collections $o$ of elements  of $C$, ordered under inclusion.  Define $o_1\perp o_2$ if, and only if, for any element $e_1\in o_1$ there do not exist 
$e_2\in o_2$ such that either $e_2\preceq e_1$ or $e_1\preceq e_2$. 
Then $\perp$ is a causal disjointness relation on $K$.   
\end{remark}

A \emph{morphism}  from a causal poset $K$ to a causal poset $K'$  is a function 
$\psi: K\to K'$ preserving both the order relation and the causal disjointness relation i.e.
\begin{equation}
\label{Aa:3}
a \leq \tilde a \  \Rightarrow \  \psi(a) \leq' \psi(\tilde a  ) 
\ \  ,  \ \   
o\perp \tilde o \ \Leftrightarrow  \ \psi(o)\perp' \psi(\tilde o) \ ,
\end{equation}
where $\leq '$ and $\perp'$ denote, respectively, the order relation and the causal 
disjointness relation of $K'$.  A \emph{symmetry  group} of a causal poset is noting but 
a subgroup  of the automorphism group of a causal poset.  \\
\indent Note that there is an evident categorical structure underlying these definitions: taking 
causal posets $K$ as objects and  the corresponding morphisms $\psi:K\to K'$ as arrows  $(K,K')$  
we get a category \textbf{Causet} that we call 
\emph{the category of causal posets}.  In particular, symmetry groups of $K$ are, by definition, subgroups of $(K,K)$.

\subsubsection{Some causal posets on spacetimes for AQFT}
\label{Aaa}

We give a non-exhaustive list of causal posets associated to different spacetimes 
which are of importance in AQFT.  Our main aim in the following sections
shall be to give  constructions that result to be well behaved on such posets.  \smallskip

The poset $K$ used  as the set of indices of the net of $\rC^*$-algebras in AQFT is a family of subsets of the spacetime
ordered under inclusion. This family is chosen to best fit the topological, causal and the symmetry, if there are any,  properties
of the spacetime.  To this end,  some minimal requests  are the following:  
1) $K$ is a base for the topology of the spacetime whose elements are connected and simply connected subsets of 
the spacetime; 2) the causal complements of elements of $K$ are connected;  3)  if $S$ is a symmetry group 
of the spacetime (a group of global isometries, or conformal diffeomorphisms for instance),  then 
$K$ is  stable under the action of $S$, i.e.\  $s(o)\in K$ 
for any $o \in K$ and $s$ in $S$, where $s(o)$ is the image of $o$ by the map $s$. In this way $S$ is realized as 
a subgroup of $(K,K)$. \smallskip

We introduce some causal posets associated to   spacetimes  which are of relevance in  AQFT. \smallskip 

\noindent 1. \emph{Globally hyperbolic spacetimes.}  Let $\cM$ be an arbitrary globally hyperbolic spacetime,   
the spacetime used in the setting of quantum fields in a gravitational background field \cite{Wal,GLRV}.     
In such a spacetime two regions are causally disjoints if they cannot be joined by any non-degenerate causal curve.  
A good choice for the set of indices is the set of \emph{diamonds} of $\cM$:  roughly,  a  diamond is   
the open causal completion of a suitable   subset  
laying on a Cauchy surface of $\cM$, see \cite{Ruz,BR} for details.  The set of diamonds is stable under any 
global symmetry of the spacetime, if there are any. \smallskip

\noindent 2. \emph{The Minkowski spacetime.}  A particular globally hyperbolic spacetime is the Minkowski spacetime $\bM^4$, the spacetime
used in the setting of relativistic quantum field theory.  
The spacelike separation is the causal disjointness relation and the symmetry group  is  the 
Poincar\'e group $\sP^\uparrow_+$.  The set of indices used in such theories is that of
 \emph{double cones} of $\bM^4$, a family smaller than that of diamonds of $\bM^4$ (see \cite{Haa}).  
These may be  defined as follows. 
Given a reference frame of the Minkowski spacetime,  let $B_R$ be  the open ball  with radius $R>0$ laying on the subspace at time $t=0$ and   centred in the origin of the reference frame. The open causal completion of $B_R$, we shall denote by $o_R$, is called a double cone with base $B_R$.  Any other double cone is obtained by letting act  the Poincar\'e group 
on the double cones $o_R$ for any  $R>0$.  For later applications, we note that the Poincar\'e group does not act freely on the set of double cones. It is easily 
seen that  the stability group  of the double cone $o_R$ defined before is the  subgroup $SO(3)$ of $\sP^\uparrow_+$ 
of spatial rotations;  so the stability group of a generic double cone is isomorphic to $SO(3)$.   \smallskip

\noindent 3. \emph{The circle $S^1$.}  This is the space used in the setting of chiral conformal quantum field theory,  
see \cite{K} for a friendly introduction. 
A symmetry group is that of M\"obius group (or the larger group of the diffeomorphisms of $S^1$)  and the causal disjointness relation is the set-disjointness. 
In this case the set of indices is the set of non-empty open intervals of $S^1$ having  a proper closure.  
For later application,  we observe that the M\"obius group does not act freely 
on intervals of $S^1$. In particular, the stabilizer of an interval $o$,  which is also 
 the stabilizer of its  open complement $S^1\setminus cl(o)$, is isomorphic to the modular group.\smallskip

The causal posets listed above, although defined on different spacetimes, 
share some important properties (see the references quoted above): 
\begin{enumerate}   
\item For any $o$ there exist $a,x$ such that 
          $a\Subset o\Subset x$. 
\item If $a\Subset o$, then exists $\tilde a$ such that $a\Subset \tilde a\Subset o$.
\item For any $o$ there exists  $a$ such that $cl(a)\perp o$. 
\item If  $cl(a)\perp o$ there exists $x$ such that $a\Subset x$ and $cl(x)\perp o$.
\end{enumerate}
Here $cl(a)$ stands for the closure of $a$ in the spacetime topology, and $a\Subset o$ stands for $cl(a)\subset o$. 
Note that any such a poset is infinite non-countable. Furthermore, as can be easily deduced by the above relations, 
any such a poset has neither maximal nor minimal elements.

%
%
%
%
\subsection{Causal functors}
\label{Ab}
Roughly speaking a {\em causal functor} is 
a functor from \textbf{Causet} to a target 
category in which causality is realized in terms of commutation relations. 
For our purposes, natural choices of the target category shall be the ones of groups 
and of $\rC^*$-algebras; in the $\rC^*$-case, the notion of causal functor  generalizes  
the one of  locally covariant quantum field theory (\cite{BFV}).\smallskip

As a preliminary step,  we introduce target categories  of the causal functor, 
generically denoted by $\boldsymbol{\rT}$: 
\begin{itemize}
\item[(i)]  the category \textbf{Grp} whose objects are groups and whose arrows 
            are group monomorphisms;
\item[(ii)] the category $\boldsymbol{\rC^*}$\textbf{Alg} whose objects 
            are unital $\rC^*$-algebras and whose arrows are unital $^*$-monomorphisms.              
\end{itemize}
We now are ready to give the definition of a causal functor. 
\begin{definition}
\label{Ab:2}
A  \textbf{causal functor} is a functor  
$\cA : \boldsymbol{\mathrm{Causet}} \to \boldsymbol{\rT}$
fulfilling the following property:  given $\psi_i\in(K_i,K)$, $i=1,2$, if there are    $o_1,o_2 \in K$ such that 
\begin{equation}
\label{Ab:2a}
o_1\perp o_2 \ \ \  and \ \ \  
\psi_1(K_1)\leq o_1 \ , \ \psi_2(K_2)\leq o_2  \ ,
\end{equation}
then 
%
%
\begin{equation}
\label{Ab:3}
\cA_{\psi_1}(A)\,\cA_{\psi_2}(B) = \cA_{\psi_2}(B)\,\cA_{\psi_1}(A) \ , 
\end{equation} 
for any $A\in \cA(K_1)$ and  $B\in \cA(K_2)$.  
When $\boldsymbol{\rT}$ is $\boldsymbol{\mathrm{Grp}}$ ($\boldsymbol{\rC^*\mathrm{Alg}}$) we shall say that $\cA$ is 
a group- ($\rC^*$-) causal functor.
%
%
%
%
%
%
\end{definition}
\noindent Some comments about the above definition. \emph{First},  the functor  $\cA$ is denoted by 
  $\cA(K)$, $K \in {\bf obj}({\bf Causet})$, at the level of objects and by
  $\cA_\psi$, $\psi \in {\bf arr}({\bf Causet})$ at the level of arrows. \emph{Secondly},  $\psi_i(K_i)\leq o_i$ means that 
any element of $\psi_i(K_i)$ is smaller that $o_i$, for $i=1,2$. So if two posets 
have causally disjoint embeddings in a larger poset, 
then the corresponding images, \emph{via} the causal functor,  commute elementwise. \smallskip

%

The functor assigning the full group  $\rC^*$-algebra to any discrete group allows us 
to assign a $\rC^*$-causal functor to any group causal functor. To this
end a key r\^ole is played by the discreteness of groups involed.  Very briefly, 
given the discrete group $G$ we denote the convolution algebra by $\ell^1(G)$ and 
its enveloping $\rC^*$-algebra by $\rC^* G$. A dense $^*$-algebra of $\ell^1(G)$ is
given by the set $\mathbb{C}G$ of functions
$f:G\to\mathbb{C}$,
which can be expressed as linear combinations 
$f=\sum_g f(g) \, \delta_g$
of delta functions. Given the discrete group morphism 
$\si:G\to H$ 
there are natural induced morphisms
\[
\tilde\si: \mathbb{C}G\to \mathbb{C}H
\ \ , \ \
\si_*:\rC^*G\to\rC^*H
\ ,
\]
defining the functor 
\[
\rC^*:\textbf{Grp}\to \boldsymbol{\mathrm{C}^*}\textbf{Alg}
\ \ , \ \
G \mapsto \rC^*G \ , \ \si \mapsto \si_* \ .
\]
A crucial point is that $\sigma_*$ is injective for any monomorphism $\sigma$;
this is not ensured when $G$ is not discrete (for details see the discussion in \cite[\S 3.2.2]{RV}).


\begin{remark}
\label{Ab:4}
Two observations are in order.\\[5pt]
1. There is a bijective correspondence between unitary representations $\rho$ of $G$ and 
representations $\pi$ of $\rC^*G$. In fact 
$\tilde\rho(f):=\sum f(g) \rho(g)$ extends to a continuous representation 
of the algebra $\ell^1(G)$. This in turn induces a representation $\rho_*$ of $\rC^*G$. 
Conversely, note that there is a group monomorphism $G\in g\to \delta_{*,g} \in \rC^*G$, 
where $\delta_{*,g}$ is the embedding of $\delta_g$ into $\rC^*G$. Therefore if $\pi$ 
is a representation of $\rC^*(G)$, then $\pi^*(g):=\pi(\delta_{*,g})$ is a unitary representation 
of $G$. These two procedures are  each other's inverses. \\[5pt]
2. Given  injective  group morphisms $\si_i:G_i\to H$, $i=1,2$, such that the groups 
$\si_1(G_1)$ and $\si_2(G_2)$ commute elementwise within $H$, then the $\rC^*$-algebras 
$\si_{*,1}(\rC^*G_1)$ and $\si_{*,2}(\rC^*G_2)$ commute elementwise within $\rC^*H$.
This is quite evident in terms of the delta functions: given 
$\delta_{g_i}\in\bC G_i$ for $i=1,2$, then 
$\tilde\si_1(\delta_{g_1})*\tilde\si_2(\delta_{g_2}) =
 \delta_{\si_1(g_1)\si_2(g_2)}= 
 \delta_{\si_2(g_2)\si_1(g_1)}=
 \tilde\si_2(\delta_{g_2})*\tilde\si_1(\delta_{g_1})$ 
and the general proof follows by density. 
\end{remark}
On these grounds, if $\cG$ is a group-causal functor, then the composition 
\begin{equation}
\label{Ab:5}
\cG_*  := \rC^*\circ \cG : \textbf{Causet}\to \rC^*\textbf{Alg} 
\end{equation}
is a $\rC^*$-causal functor. Covariance is obvious while causality follows from 
the previous observation.  We shall call {\em group $\rC^*$-causal functor} such a causal functor.

\begin{remark}
\label{Ab:6}
Any group $\rC^*$-causal functor $\cG_*$ has a nontrivial cyclic representation. Let $\pi_K$ be the left regular representation 
of $\cG(K)$ on the Hilbert space $\ell^2(\cG(K))$ and denote its extension to  the $\rC^*$-algebra $\rC^*\cG(K)$  
by $\pi_{*,K}$.   Given $\psi\in(K,K')$,  let $V_\psi:\ell^2_{K} \to \ell^2_{K'}$ be the isometry defined by 
$V_{\psi}(\delta_{g}):= \delta_{\psi(g)}$,  for any  $g\in \cG_(K)$. Then one can easily see that  
\begin{equation}
\label{Ab:7}
V_\psi\, \pi_{*,K}= \pi_{*, K'}\circ\cG_{*,\psi} \, V_{\psi} \ \ , \ \ 
V_{\psi' \circ \psi} = V_{\psi'} V_\psi \ ,  \qquad  \forall \psi \in (K,K') \ , \ \psi' \in(K',K'')  \ .  
\end{equation}
So, the pair $(\pi_*,V)$ is a representation of $\cG_*$ 
\footnote{This is a slight generalization of the notion of representation used in \cite{BR,RV},
where unitarity of the analogues of the $V_\psi$ is required.}.
Moreover, let $\delta_{1_K}$ be the delta function of the identity $1_K$ of $\cG_K$. 
Since $\psi_{1_{K_1}} = 1_{K_2}$ we have, according to the definition of $V_{\psi}$, 
that $ V_{\psi}\, \delta_{1_{K_1}}=\delta_{1_{K_2}} $, 
hence the collection $\delta_{1_K}\in\ell^2(\cG_K)$, for $K\in\textbf{Causets}$,  
is a cyclic  invariant vector.  
\end{remark}

\subsubsection{A comment on the causality condition (\ref{Ab:2a}) and the relation  with locally covariant QFT} 
\label{Ac.c}

We now show a relation between   $\rC^*$-causal functors and the locally covariant quantum field theory.  
We refer the reader to the original paper \cite{BFV} for a complete description  and 
physical motivation of the theory. We also recall some new developments of the theory: 
\cite{Pi} for  a generalization to the conformally covariant case and  \cite{Few, FV}  for an analysis of 
theories  describing  the same physics in all spacetimes,  the so called SPASs condition. \smallskip

We briefly recall the axioms of a  locally covariant quantum field theory.
To this end, let  $\textbf{Loc}$  denote  the category 
whose objects are 4-d globally hyperbolic spacetimes $\cM$
and whose arrows  $\psi:\cM_1\to \cM_2$ are isometric embeddings, 
preserving the orientation and the time orientation, and such that 
if $x,y\in \psi(\cM_2)$, then $J^+(x)\cap J^-(y)\subseteq \psi(\cM_2)$,  where 
$J^{+/-}(x)$ is the causal future/past of the point $x$.  \\
\indent  A \emph{locally covariant quantum field theory}  
is nothing but a functor  $\sA: \textbf{Loc}\to\boldsymbol{\rC^*}$\textbf{Alg}. We say that $\sA$ 
is \emph{causal} whenever  given   $\psi_i:\cM_i\to\cM$, $i=1,2$, such that  
$\psi_1(M_1)$ and $\psi_2(M_2)$ are causally disjoint in $\cM$,  then 
\[
\sA_{\psi_1}(A)\, \sA_{\psi_2}(B) =  \sA_{\psi_2}(B) \, \sA_{\psi_1}(A) \ ,  
\]
for any $A\in \sA(\cM_1)$ and $B\in \sA(\cM_2)$; $\sA$ is said to satisfy the \emph{time-slice axiom} whenever 
\[
\sA_\psi(\sA(\cM_1)) =\sA(\cM_2)
\] 
for any $\psi:\cM_1\to \cM_2$ such that $\psi(\cM_1)$ contains a Cauchy surface of $\cM_2$. \smallskip


In order to get a connection between causal functors and the locally covariant quantum field theory (which are causal) we must relax the causality condition (\ref{Ab:2a}). Let $\cA':\boldsymbol{\mathrm{Causet}}\to \boldsymbol{\rC^*\mathrm{Alg}}$ be a functor 
such that  given $\psi_i\in(K_i,K)$ for $i=1,2$, if 
\begin{equation}
\label{Ac:1}
\psi_1(K_1)\perp \psi_2(K_2)  \ , 
\end{equation} 
then (\ref{Ab:3}) holds, where $\psi_1(K_1)\perp \psi_2(K_2)$ means that $\psi_1(o_1)$ is causally disjoint from $\psi_2(o_2)$
for any $o_1\in K_1$, $o_2\in K_2$. This clearly ensures that $\cA'$ is actually a causal functor.  
Let $\rK$ denote the functor assigning to any globally hyperbolic spacetime $\cM$ the set of diamonds $K(\cM)$  of $\cM$ ordered 
under inclusion, and to any  arrow $\psi:\cM_1\to \cM_2$  the induced map  $K(\cM_1)\ni o \to \psi(o)\in K(\cM_2)$
(diamonds are stable under the isometric embeddings). Then, using (\ref{Ac:1}), it is easily seen that 
\[
\sA:\cA'\circ \rK : \boldsymbol{\mathrm{Loc}} \to \boldsymbol{\rC^*\mathrm{Alg}} \ , 
\]
is a causal locally covariant quantum field theory. 
Obviously, because of the generality of the notion of causal functor, nothing can be said about the time-slice axiom.  \smallskip

The causality relation  (\ref{Ab:2a}),  which actually  characterizes the notion of a causal functor  and results to be 
too restrictive for locally covariant quantum field theory, 
is mainly motivated  by  
gauge theories and  in particular by a remark due to D.Buchholz  and J.E. Roberts 
in the setting of quantum electrodynamics (see \cite{Rob0} for details).
In fact, assume that the vector potential $A_\mu$  is realized as a quantum field in the Minkowski spacetime. 
It can be easily seen that  a suitable smoothing of $A_\mu$ around a closed curve is an \emph{observable} 
since it is related to the electromagnetic field \emph{via} the Stokes theorem, similarly to the classical case.  Clearly, this observable is not localized around the curve but it can be localized around any surface having that curve as a boundary. Consequently,  if we take two causally disjoint curves $\gamma_1$ and $\gamma_2$ \emph{forming a link}, the smoothings of $A_\mu$ around these two curves do not commute. The causality relation 
 (\ref{Ab:2a}) forbids this situation. Actually, as observed in Section \ref{Aaa}, the posets  that we use as indices for nets over spacetimes have connected and \emph{simply connected} regions as elements, so two curves localized within 
two causally disjoint regions cannot form a link.  \smallskip

\subsection{Nets of $\rC^*$-algebras} 
\label{Ac}

Mimicking the locally covariant quantum field
theory, we now show that once a $\rC^*$-causal functor $\cA$ is given,  
it is possible to construct a causal net of $\rC^*$-algebras over any causal poset $K$ which is, 
in particular, covariant with respect to any symmetry group of $K$.
\smallskip

Let us first recall some basic definitions.  
A net of $\rC^*$-algebras  over $K$ is a pair $(\cA,\jmath)_K$ where $\cA$ is a correspondence 
$o \mapsto \cA_o$, $o\in K$,  of unital $\rC^*$-algebras, the \emph{fibres} of the net,  
and $\jmath$ is a correspondence $o\leq a \mapsto \jmath_{oa}$ 
of unital $^*$-monomorphisms $\jmath_{ao}:\cA_o\to\cA_a$, the \emph{inclusion maps}, satisfying 
the \emph{net relations}
\begin{equation}
\jmath_{o''o'}\circ\jmath_{o'o}=\jmath_{o''o} \ , \qquad o''\geq o'\geq  o \ . 
\end{equation}
The net is said to be \emph{causal} whenever 
\begin{equation}
 [\jmath_{ao_1}(\cA_{o_1}),\jmath_{ao_2}(\cA_{o_2})]=0 
  \ , \qquad   
 a\geq  o_1 , o_2   
 \ , \
 o_1\perp o_2 \ .  
\end{equation}
If $G$ is a symmetry group of $K$ the net  is said to be {\em $G$-covariant} whenever 
for any $o \in K$ and $g\in G$ there is an isomorphism $\alpha^o_g:\cA_o\to \cA_{go}$  
subjected to the following conditions: 
\begin{equation}
 \alpha^{ho}_g\circ\alpha^o_h= \alpha^o_{gh}
\  \ , \ \ 
 \alpha^a_g\circ\jmath_{ao}= \jmath_{ga\,go}\circ \alpha^o_g
\  , \qquad  
  a\geq o   \ , \ h,g\in G \ .   
\end{equation}
We now show how a net of $\rC^*$-algebras is defined from a $\rC^*$-causal functor 
$\cA$ after fixing a causal poset $K$.
As a first step, for any $o\in K$  we define the set 
\[
(K|o):= \{ a\in K \,|\, a\leq o\} \ ;
\]
this, equipped with the ambient order relation $\leq$ and the ambient causal disjointness 
relation $\perp$, is a causal poset too, so the inclusion 
$\io_{K,(K|o)}:(K|o)\to K$ is an arrow of $((K|o),K)$. 
We define
\begin{equation}
\label{Ac:2}
 \cA_o := \cA_{\io_{K,(K|o)}}(\cA(K|o)) \ , \qquad  o\in K \ .
\end{equation}
Any $\cA_o$  is a $\rC^*$-subalgebra of $\cA(K)$ and 
\begin{equation}
\label{Ac:3}
\cA_o \subseteq \cA_a
\ , \qquad   o \leq a \ .
\end{equation}
To prove (\ref{Ac:3}),  observe that according to the definition  (\ref{Ac:2}) the inclusion   of 
$\cA_o$ into  $\cA_a$ is given by 
$\cA_{\io_{K,(K|a)}}\circ \cA_{\io_{(K|a),(K|o)}}\circ \cA^{-1}_{\io_{K,(K|o)}}$ where 
$\io_{(K|a),(K|o)}\in ((K|o),(K|a))$. 
Using functoriality, we have that
\begin{align*}
\cA_{\io_{K,(K|a)}}\circ \cA_{\io_{(K|a),(K|o)}}\circ \cA^{-1}_{\io_{K,(K|o)}} & 
= \cA_{\io_{K,(K|a)}\circ\io_{(K|a),(K|o)}}\circ \cA^{-1}_{\io_{K,(K|o)}} \\
&  = \cA_{\io_{K,(K|o)}}\circ \cA^{-1}_{\io_{K,(K|o)}} = \cA_{\io_{K,(K|o)}}\circ \cA^{-1}_{\io_{K,(K|o)}} \\ 
& = \mathrm{id}_{\cA(K)} \ , 
\end{align*}
which gives the  desired inclusion (\ref{Ac:3}), that we denote by 
$\mathrm{id}_{ao} : \cA_o \to \cA_a$.
This yields the net of $\rC^*$-algebras $(\cA,\mathrm{id})_K$.
To prove causality we note that, since the image of $\io_{K,(K|o)}$ is $(K|o) \leq o$, and since $\cA$ is causal, 
we have, as desired,
\begin{equation}
\label{Ac:4}
 [\cA_o,\cA_a]=0 \ , \qquad  o\perp a \ .
\end{equation}
Finally, let $S\subseteq (K,K)$ be a symmetry group of $K$. Define 
\begin{equation}
\alpha_s:= \cA_s \ , \qquad s\in S \ . 
\end{equation} 
Then $\alpha_s:\cA(K) \to \cA(K)$ is a $^*$-automorphism for any $s\in S$ and
$\io_{K,(K|go)}^{-1} \circ s\circ\io_{K,(K|o)} : (K|o) \to (K|so)$ 
is an isomorphism of causal posets, therefore $\alpha_s:\cA_o\to\cA_{so}$ is a $^*$-isomorphism.
Since the inclusion maps are constant, it easily follows that $(\cA,\mathrm{id},\alpha)_K$
is a $S$-covariant causal net of $\rC^*$-algebras over $K$. Summing up:
\begin{theorem}
\label{Ac:5}
Given a $\rC^*$-causal functor $\cA$, for any causal poset $K$ we have:
\begin{itemize}
\item[(i)]  $(\cA,\mathrm{id})_K$ is a causal net of $\rC^*$-subalgebras of $\cA(K)$; 
\item[(ii)] If $S$ is a symmetry group of $K$, then $(\cA,\mathrm{id},\alpha)_K$ is $S$-covariant.  
\end{itemize}
\end{theorem}


\section{Causal functors from groups of loops}
\label{B}
In this section we construct a causal functor using the structure of causal posets only. 
We first associate a (free) group  with any poset and show that this association 
yields a functor from the category of causal posets to groups. Afterwards we introduce causal commutators. 
These are normal subgroups of the free group.  We then perform the quotient and show that this
leads to a group-causal functor hence, according to the results of the previous section, to a $\rC^*$-causal functor.

\subsection{Simplices}
\label{Ba}
A simplicial set $\Si_*(K)$ can be defined over any poset $K$,
in the following way  \cite{Rob,RR}. As a first step we set $\Si_0(K) := K$ and then,
inductively, we define elements of $\Si_n(K)$, the $n$-{\em simplices}, 
as follows,
\[
x := (|x|; \partial_0x , \ldots , \partial_nx )
\ \ : \ \
\left\{
\begin{array}{ll}
\partial_ix \in \Si_{n-1}(K) \ , \ |x| \in K  \ , \ \forall i = 0 , \ldots , n \ , \ n \geq 1 \ ,
\\
| \partial_ix | \leq |x| \ , \ \forall i = 0 , \ldots , n \ , \ n \geq 1 \ ,
\\
\partial_{ij}x := \partial_i(\partial_jx) = \partial_{j,i+1}x 
\ , \ 
\forall i \geq j 
\ , \ 
n \geq 2 \ .
\end{array}
\right.
\]
The $(n-1)$-simplices $\partial_ix$ are called the {\em faces}, whilst
$|x|$ is called the {\em support} of $x$. Note that the previous definition 
also yields the face maps 
$\partial_i : \Si_n(K) \to \Si_{n-1}(K)$,
$i = 0, \ldots , n$,
$n \geq 1$.
There are also \emph{degeneracies operators} $\si_i: \Si_n(K) \to \Si_{n+1}(K)$,
$i = 0, \ldots , n$,
$n \geq 0$ (see the cited references for the definition), making $\Si_*(K)$ a simplicial set.\\ 
\indent The simplicial set $\Si_*(K)$ is symmetric since it is closed under any permutation 
of the vertices of any of its simplex. Anyway since in what follows we shall mainly use 
$1$-simplices of $\Si_*(K)$, we focus on the properties of $\Si_1(K)$. Any 1-simplex $b$ is formed by a triple 
of elements of $K$,
\[
b = ( |b|;\partial_0b , \partial_1b)
\ \ : \ \
\partial_0b , \partial_1b \leq |b|
\ .
\]
The intuitive content of the 1--simplex $b$ is that it is a segment starting at the ``point" $\partial_1b$
and ending at $\partial_0b$. A \emph{degenerate} 1-simplex of the form $\si_0a$ for $a\in K$, is 
the 1-simplex having the support and faces equal to $a$.

Since $\Si_*(K)$ is symmetric we can associate to any 1-simplex $b$ 
its opposite $\overline{b}$: this is the 1-simplex having 
the same support as $b$ and inverted faces: 
\[
|\overline{b}|:=|b| \ , \ \ \partial_1\overline{b}:=\partial_0b \ , \ \ 
 \partial_0\overline{b}:=\partial_1b \ . 
\]
Another simplicial sets which shall play a r\^ole in what follows is the \emph{nerve} 
of $K$. This is the subsimplicial set $\rN_*(K)$ of those  simplices of $\Si_*(K)$ 
whose vertices are totally ordered and whose support coincide with the greatest vertex. 
In particular for 1-simplices we have that 
$b\in\rN_1(K)$ if, and only if, $\partial_1b\leq \partial_0b=|b|$.


\subsection{The free group of loops on a poset} 
\label{Bb}

A first step toward the construction of a causal functor is to associate a free group and its subgroup of loops
to a suited subset of the set of 1-simplices of a poset. The subgroup of loops will play 
a more relevant r\^ole  concerning the construction of the target category of the functor.
As a reference  for free group theory  see \cite{Joh}.  \smallskip

Given a causal poset $K$ we define
\footnote{We are following a different idea from that used 
in \cite{RV} to define causal net of $\rC^*$-algebras. This because 
we want to remove any topological effect 
from the definition of the group associated with a causal poset. In this regard, we  recall that the  
first homotopy group of a poset is encoded 
in the nerve $\rN_1(K)$, see \cite{RRV}.}\smallskip  
\begin{equation}
\label{Bb:1}
 \rT_1(K):=\Si_1(K)\setminus\rN_1(K) \  .  
\end{equation}
So $\rT_1(K)$ is the set of 1-simplices of $K$ which do not belong to the nerve of $K$, or, 
equivalently,  $b\in\rT_1(K)$ if, and only if, $|b|>\partial_0b,\partial_1b$.  
This set is closed under the operation of making the opposite and, in general is not empty.
Moreover, it is sufficient to describe path connecteness of the poset (this will clear soon).  \\ 
\indent Now, let us consider the group generated by $\rT_1(K)$ with relation $b^{-1} = \overline{b}$, i.e., 
\begin{equation}
\label{Bb:2}
\rF(K):=\{ b\in\rT_1(K) \,|\, b^{-1}=\overline{b}\} \ .   
\end{equation}
$\rF(K)$ is (non-canonically) isomorphic to a free group. To this end let $\sim_{opp}$ denote the equivalence relation
$b \sim_{opp} b' \Leftrightarrow \ b' = \overline{b}$ for any $b \in \rT_1(K)$. 
Then by the axiom of choice there is a section 
$s : \rT_1(K) / \sim_{opp} \ \longrightarrow \ \rT_1(K)$
and it is easily seen that $\rF(K)$ is isomorphic to the free group
generated by $s(\rT_1(K) / \sim_{opp})$. \\
\indent Thus, elements of $\rF(K)$ are words  
\[
 w:= b_n\,b_{n-1}\,\ldots \,b_1 \ , \qquad b_i\in \rT_1(K) \ . 
\]
The inverse of a word is given by $\overline{w}:=\overline{b}_1,\ldots,\overline{b}_n$,  
where $\overline{b}$ is the opposite of the 1-simplex $b$ of $\rT_1(K)$. The empty word $1$ is the identity of the group. 
The \emph{support}  $|w|$ of a word $w$ is the collection of the supports of the generators of $w$, as a reduced word $w^\rr$.
For instance if $w=b_2 b\, \overline{b}\, b_1$ and is $b_1 \neq \overline{b}_2$ then 
$|w|=\{|b_2|,|b_1|\}$ since the  reduced word $w^{\rr}$  equals  $b_2 b_1$.
A word $w=b_n\,\ldots\,b_1$ is said to be a \emph{path} whenever its generators satisfy the relation 
\begin{equation}
\label{Bb:3}
 \partial_{0}b_{i+1}=\partial_1b_{i} \ , \qquad i= 1,\ldots, n-1 \ . 
\end{equation}
We set $\partial_1w:=\partial_1b_1$ and  $\partial_0w:=\partial_0b_n$ and call 
these 0-simplices, respectively, the \emph{starting} and the \emph{ending point} of the path $w$. 
We shall also use the notation $w:a\to o$ to denote a path from $a$ to $o$. 
Finally, a path $w:o\to o$ is said to be a \emph{loop over $o$}.  Note that if $w:a\to o$ is a path, then the  reduced word $w^{\rr}$
is still a path from  $a$ to $o$.\\ 
\indent  Causality is extended  from $K$ to $\rF(K)$ as follows: two words 
$w_1$ and $w_2$ are \emph{causally disjoint}, 
\begin{equation}
\label{Bb:4}
w_1\perp w_2 \ \ \iff \ \ \exists o_1,o_2\in K \ , \  |w_1|\leq o_1 \ , \ |w_2|\leq o_2 \ , \ o_1\perp o_2 \ .  
\end{equation}
We stress that  the above relation is given on the reduced words   $w^\rr_1$ and $w^\rr_2$   according to the definition 
of the support of a word. \smallskip 

We now show that the correspondence associating $K\mapsto\rF(K)$ is a functor.
Let $K_1,K_2\in\textbf{Causets}$ and $\psi \in (K_1,K_2)$. Extend the action 
of $\psi$ from the poset $K_1$  to $\rT_1(K_1)$ by setting $\psi(b)$, by a little abuse of notation, as 
\begin{equation}
\label{Bb:5}
|\psi(b)|:= \psi(|b|) \ \ , \ \  \partial_1\psi(b):= \psi(\partial_1b) \ \ , \ \ \partial_0\psi(b):= 
\psi(\partial_0b) \ .  
\end{equation}
Since $\psi$ is order preserving an injective, $\psi(b)$ is a 1-simplex of 
$\rT_1(K_2)$. Moreover,  $\overline{\psi(b)}=\psi(\overline{b})$. Finally if 
$w=b_n\,\cdots \,b_1$ is any word of $\rF_{K_1}$ we have that 
\begin{equation}
\label{Bb:6}
\psi(w):=\psi(b_n)\,\cdots\,\psi(b_1) 
\end{equation}
is a word of $\rF(K_2)$. Clearly, if $w$ is a reduced word of $\rF(K_1)$, 
then $\psi(w)$ is a reduced word of $\rF(K_2)$; if $w:o\to o $ is a loop of $\rF(K_1)$, 
then $\psi(w):\psi(o)\to\psi(o)$ 
is a loop of $\rF(K_2)$; if $w\perp_1 v$ are causally disjoint words in $\rF(K_1)$, 
then $\psi(w)\perp_2\psi(v)$ in $\rF(K_2)$.\\
\indent According to the above definition for any morphism $\psi:K_1\to K_2$, we have that 
$\psi: \rF(K_1)\to \rF(K_2)$ is an injective group morphism as well. So
\[
 \rF: \textbf{Causet} \to \textbf{Grp}
 \ \ , \ \
 K\mapsto \rF(K) \ , \ \ 
 \psi \mapsto \psi  \ , 
\]
is a functor. We now give the following key 
\begin{definition}
\label{Bb:7}
We call the \textbf{group of loops} of $K$ the subgroup  $\rL(K)$  of $\rF(K)$ generated by loops. 
\end{definition}
\noindent In general $\rL(K)$ is not a normal subgroup of $\rF(K)$. 
An element of $\rL(K)$ is, by definition, a word of the form $w=p_n\,p_{n-1}\,\cdots\, p_1$ 
where any $p_i$ is a loop over a 0-simplex $o_i$. 
The reduced word $w^\rr$  of $w$ may have the following forms: 
if  $o_i\ne o_{i+1}$ for any $i$, since no cancellation is possible between loops defined 
over different base point, then $w^{\rr}=p^{\rr}_n\, p^{\rr}_{n-1}\,\cdots\, p^{\rr}_1$ 
where $p^{\rr}_i$ is the reduced word of $p_i$. 
If, for instance $o_1=o_2$,  $o_i\ne o_{i+1}$ for $i\geq 2$,  since $p_2p_1$ 
is a loop over $o_1$, then 
$w^{\rr}=p^{\rr}_n\, p^{\rr}_{n-1}\,\cdots\, p^{\rr}_3 \, p^\rr_{21}$ 
where $p^\rr_{21}$ is a loop obtained by reducing the loop $p_2p_1$. Hence  $\rL(K)$ is stable under reduction of words.\\
\indent Since, as observed before,   any $\psi\in(K_1,K_2)$  preserves loops,  
we have that 
$\psi : \rL(K_1)\to\rL(K_2)$, 
implying, as before,   that  the assignment
\[
\rL: \textbf{Causet} \to \textbf{Grp}  \ \ , \ \
 K\mapsto \rL(K) \ , \ \ 
 \psi \mapsto \psi  \ , 
\]
is a functor.\\ 
%



\subsection{Causality} 
\label{Bc}

The final step toward the construction of a causal functor is to introduce causality. This will be done 
in terms of commutators of the groups introduced in the previous section. 
We shall see that the quotient of the group of loops by  commutators gives rise to a causal functor. \smallskip
%
%
To start, we denote  as usual the commutator of two words $w_1,w_2$ of $\rF(K)$ by   $[w_1,w_2]=w_1w_2\overline{w}_1\overline{w}_2$.  
\begin{definition}
\label{Bc:1}
A \textbf{causal commutator of $\rF(K)$} is a word of the form 
\[
 w\,[p_1,p_2]\,\overline{w} \ \ \ \mbox{where}  \ \ \  p_i:o_i\to o_i\ , i=1,2\  \ , \ \ p_1\perp p_2\ \  , \ \ w\in\rF(K) \ ,
\] 
and it is a \textbf{causal commutator of $\rL(K)$}  when $w\in \rL(K)$.  
We  denote the group  generated by 
causal commutators of   $\rF(K)$ and $\rL(K)$ respectively by   $\rC\rF(K)$ and $\rC\rL(K)$. 
\end{definition}
\noindent Note that, by definition,  $\rC\rF(K)$ is  a normal  subgroup  of $\rF(K)$  and   $\rC\rL(K)$ is  a normal  subgroup  of $\rL(K)$. 
Moreover,  $\rC\rL(K)$ is a subgroup, in general not normal, of $\rC\rF(K)$.  An important  relation between these two groups 
is given by the following 
 %
%
%
%
%
%
%
%
%
\begin{lemma}
\label{Bc:2}
$\rC \rL(K)= \rC \rF(K)\cap \rL(K)$ for any causal poset $K$.    
\end{lemma}
\begin{proof}
Clearly $\rC\rL(K)\subseteq\rC \rF(K)\cap \rL(K)$.  So, consider a generic element $C$ of $\rC \rF(K)$. By Definition \ref{Bc:1}, 
$C$ has the form
\[
C=w_1\,[q_1,p_1]\,\overline{w}_1 \,  w_2\,[q_2,p_2]\,\overline{w}_2 \cdots  w_n\,[q_n,p_n]\,\overline{w}_n \ ,      
\]
where $q_i:o_i\to o_i$ and $p_i:a_i\to a_i$ are causally disjoint loops for 
$i=1,\ldots,n$ and $w_i\in\rF(K)$. Assume that $C\in\rL(K)$. 
Then we must show that $w_i\in\rL(K)$ for any $i$. Assume that it is not the case. 
Then it is enough to consider the case 
where $w_n\not\in\rL(K)$. In fact if $w_n\in\rL(K)$, then $w_n\,[q_n,p_n]\,\overline{w}_n\in\rL(K)$. So 
\[
w_1\,[q_1,p_1]\,\overline{w}_1 \,  \cdots  \, w_{n-1}\,[q_{n-1},p_{n-1}]\,\overline{w}_{n-1}\in\rL(K)
\]
and we may repeat the same reasoning with respect to $w_{n-1}$. \\
\indent So assume that $w_{n}\not\in\rL(K)$. It is enough to consider the case where
$w_n$ is a path from $x_n$ to $y_n$ with $x_n\ne y_n$. Clearly   
$y_n\ne a_n$, where $a_n$ is the basepoint of $p_n$, leads to a contradiction. But also the equality
$y_n=a_n$ leads to a contradiction. In fact in this case $\overline{p}_n\overline{w}_n$
is a path from $x_n$ to $y_n$ and no reduction is possible in order to get a loop
since $\overline{q}_n$ is causally disjoint from $p_n$ (recall that causal disjointness is irreflexive, Section \ref{Aa}).   So the only possibility is that 
$w_n$ is a loop. This, and the previous observation,  complete the proof. 
\end{proof}
Again, since any $\psi\in(K_1,K_2)$  preserves inclusions, loops, and causal disjointness, 
we have that $\psi:\rC {\sharp}(K_1) \to \rC {\sharp}(K_2)$ is an injective group morphism for 
$\sharp=\rL,\rF$. So, as before,  the assignment 
\[
\rC {\sharp} : \textbf{Causet} \to \textbf{Grp}
\ \ , \ \
K\mapsto\rC {\sharp}(K) \ , \ \ \psi\mapsto \psi
\ ,
\]
with $\sharp=\rL,\rF$, is a   functor. \smallskip 

We now are in a position to define the desired causal functor. As observed $\rC \rL(K)$
and $\rC \rF(K)$ are, respectively, a normal subgroup of $\rL(K)$ and 
a normal subgroup of $\rF(K)$. So consider the quotient groups 
\begin{equation}
\label{Bc:3}
\widehat{\rF}(K) := \rF(K) / \rC \rF(K)  \ \ and  \ \ \widehat{\rL}(K):= \rL(K) / \rC \rL(K) \ ,  
\end{equation}
and denote the quotient maps, respectively, by the symbols $Q_K^{\rF}$ and $Q_K^{\rL}$.
Furthermore let 
\begin{equation}
\label{Bc:4}
 \imath_K \circ Q_K^{\rL}(w):= Q_K^{\rF}(w)   \ , \qquad w\in\rL(K) \ . 
\end{equation}
Since $\rC \rL(K) \subseteq \rC \rF(K)$ this definition is well posed, i.e. 
if $Q_K^{\rL}(w)=1$ then $Q_K^{\rF}(w)=1$, and yields   
a group monomorphism $\imath_K:\widehat{\rL}(K)\to \widehat{\rF}(K)$.  
Injectivity follows by observing that   if $Q_K^{\rF}(w)=1$, then $w\in\rC \rF(K)$.
Since $w\in\rL(K)$, by Lemma \ref{Bc:2}, we have that $w\in\rC\rL(K)$, so   $Q_K^{\rL}(w)=1$. 
Similarly, by setting for any $\psi\in (K_1,K_2)$
\begin{equation}
\label{Bc:5}
\widehat{\rF}_{\psi}\circ Q_{K_1}^\rF := Q_{K_2}^\rF\circ\psi \  \ , \ \ 
\widehat{\rL}_{\psi}\circ Q_{K_1}^\rL := Q_{K_2}^\rL\circ\psi \ , 
\end{equation}
we get that  $\widehat{\rF}_\psi:\widehat{\rF}(K_1)\to \widehat{\rF}(K_2)$ and 
$\widehat{\rL}_\psi  :\widehat{\rL}(K_1)\to \widehat{\rL}(K_2)$
are injective group morphisms. This leads to the commuting diagram  
\begin{equation}
\label{Bc:6}
\xymatrix{
\rF(K_1)\ar[rd]^{Q_{K_1}^\rF}\ar[rrr]^{\psi} & & & \rF(K_2)\ar[ld]_{{Q}_{K_2}^\rF}\\
& \widehat{\rF}(K_1)\ar[r]^{\widehat{\rF}_{\psi}} & \widehat{\rF}(K_2) & \\
& \widehat{\rL}(K_1)\ar[u]_{\imath_{K_1}}\ar[r]_{\widehat{\rL}_{\psi}} & \widehat{\rL}(K_2)\ar[u]^{\imath_{K_2}} & \\
\rL(K_1)\ar[ru]_{Q_{K_1}^\rL}\ar[rrr]^{\psi}\ar[uuu]^{\subseteq} & & & \rL(K_2)\ar[lu]^{Q_{K_2}^\rL}\ar[uuu]_{\subseteq}}\\
\end{equation}
and to the existence of a group-causal functor. 
%
%
%
%
%
%
\begin{theorem}
\label{Bc:7}
The following assertions hold. 
\begin{itemize}
\item[(i)] The assignment $\widehat\rL: \boldsymbol{\mathrm{Causet}}\to\boldsymbol{\mathrm{Grp}}$ 
is a group-causal functor endowed with the natural transformation 
$Q^\rL:\rL\to\widehat{\rL}$ defined by $K\mapsto Q^\rL_K$. 
\item[(ii)]  The assignment 
$\widehat\rF : \boldsymbol{\mathrm{Causet}}\to\boldsymbol{\mathrm{Grp}}$
is a functor endowed with a natural transformation 
$Q^\rF:\rF\to\widehat{\rF}$ defined by  $K\mapsto Q^\rF_K$.
\end{itemize}
These two functors are related by  a natural transformation $\imath : \widehat\rL \to \widehat\rF$ 
defined by $K\to\imath_K$ and  satisfying the relation $\imath\circ Q^\rL= Q^\rF$, where on the r.h.s.\ of this equation,  
$\rL$ is identified as a subfunctor of $\rF$. 
\end{theorem}
\begin{proof}
We only have to show the causal properties of $\widehat{\rL}$, that is, equation (\ref{Ab:3}).  
Consider $K_i,K\in \textbf{Causet}$ and $\psi_i\in (K_i,K)$ for $i=1,2$ and 
assume that $\psi_i(K_i)\leq o_i\in K$, for $i=1,2$, and $o_1\perp o_2$. 
Given $l_i\in \rL(K_i)$, by equation (\ref{Bc:5}) we have that 
\[
\widehat{\rL}_{\psi_i}\circ Q^{\rL}_{K_i}(l_i) = 
 Q^L_K\circ \psi_i(l_i) \ , \qquad i=1,2 \ ,
\]
and the proof follows since 
$\psi_1(l_1)$ and $\psi_2(l_2)$ are causally disjoint loops  of $\rL(K)$.  
\end{proof}
Using the functor $\rC^*$ introduced in Section \ref{Ab} and, in particular, the equation (\ref{Ab:5}),  
we derive from Theorem \ref{Bc:7},  the $\rC^*$-causal functor. 
\begin{corollary}
\label{Bc:8}
The following assertions hold.
\begin{itemize}
\item[(i)] The assignment $\widehat{\rL} _*:=\rC^*\circ \widehat{\rL}: \boldsymbol{\mathrm{Causet}}\to\boldsymbol{\rC^*\mathrm{Alg}}$   
is  a $\rC^*$-causal functor.
\item[(ii)] The assignment  $ \widehat{\rF}_*:=\rC^*\circ \widehat{\rF}: \boldsymbol{\mathrm{Causet}}\to\boldsymbol{\rC^*\mathrm{Alg}}$  is a  $\rC^*$-functor. 
\end{itemize}
There is   a natural transformation $\imath_*:=\rC^*\circ \imath: \widehat{\rL} _* \to \widehat{\rF} _*$.
\end{corollary}

\subsection{The net of causal loops}
\label{Bd}

Finally, using the $\rC^*$-causal functor  $\widehat{\rL} _*$ and the $\rC^*$-functor $\widehat{\rF} _*$  of Corollary \ref{Bc:8}, 
we arrive to the main result of this  section, i.e.\ the definition of the causal and covariant net of $\rC^*$-algebras 
associated to a causal poset.
\smallskip

\indent Let $K$ be a causal poset with a symmetry group $S$. We consider the causal and covariant net of $\rC^*$-algebras
associated with  the  $\rC^*$-causal functor 
$\widehat{\rL}_*$. Namely, remembering the construction of Section \ref{Ac},  the fibres of the net are defined  
from equation (\ref{Ac:2}) as
\begin{equation}
\label{Bd:1}
 \cA_o := \widehat{\rL}_{*,\io_{K,(K|o)}}(\widehat{\rL}_*(K|o)) \  \ , \qquad   o\in K \  .   
\end{equation} 
Correspondingly, the $\rC^*$-algebra associated with $K$ is defined as $\cA(K):= \widehat{\rL}_*(K)$;  this yields 
a causal net $(\cA,\mathrm{id})_K$  of $\rC^*$-subalgebras of $\cA(K)$. Finally,  the action of the symmetry group $S$ of $K$ 
on this net is defined by 
\begin{equation}
\label{Bd:2}
\alpha_s:= \widehat{\rL}_{*,s} \ , \qquad s\in S \ . 
\end{equation} 
Then, according to Theorem \ref{Ac:5}, $(\cA,\mathrm{id},\alpha)_K$  
 is a causal and $S$-covariant net of $\rC^*$-algebras over $K$, that we call 
\textbf{the $\rC^*$-net of causal loops over $K$}. \\
\indent In a similar way,  the $\rC^*$-functor $\widehat{\rF}_*$ yields 
the $\rC^*$-dynamical system $(\cF(K),S,\varphi)$, where $\cF(K):= \widehat{\rF}_*(K)$ and 
$\varphi_s:=\widehat{\rF}_{*,s}$ for any $s\in S$.
Then,  as the fibres   of the net of causal loops  are subalgebras  
of $\cA(K)$ and using the  natural transformation $\rho:=\imath_*$ given in Corollary \ref{Bc:8},  we have that 
\begin{equation}
\label{Bc:9}
(\cA,\mathrm{id},\alpha)_K\stackrel{\subseteq}{\longrightarrow} (\cA(K),S,\alpha) \stackrel{\rho_K} {\longrightarrow}(\cF(K),S,\varphi) 
\end{equation}
is a sequence of $S$-covariant monomorphisms.  These observations lead to the following
\begin{corollary}
\label{Bc:10}
Let $K$ be a causal poset with a symmetry group $S$. Then, 
\begin{itemize}
\item[(i)]  The net  of causal loops  $(\cA,\mathrm{id},\alpha)_K$ is a causal 
and $S$-covariant net of $\rC^*$-subalgebras of $\cA(K)$. 
\item[(ii)] $\rho_K: (\cA(K),S,\alpha)\to (\cF(K),S,\varphi)$ is an $S$-covariant monomorphism of $\rC^*$-dynamical systems.  
\end{itemize}
\end{corollary}
\noindent This corollary represents  one of the main results of the present paper.  In fact, 
for the class of spacetimes discussed in 
Section \ref{Aaa},   
the corresponding net of causal loops  is a causal net of non-Abelian $\rC^*$-algebras which is covariant under the symmetry group of the spacetime. 
\begin{theorem}
\label{Bc:11}
The following assertions hold: 
\begin{itemize}
\item[(i)]  the net of causal loops over the set of \textbf{double cones} of the Minkowski spacetime $\bM^4$
is covariant under the Poincar\'e group; 
\item[(ii)] the causal net of loops over the set of \textbf{open intervals} 
of $S^1$ is covariant under the M\"obius group (or the diffeomorphism group of $S^1$); 
\item[(iii)]  the causal net of loops over the set of \textbf{diamonds} of a globally hyperbolic spacetime $\cM$ 
        is covariant under any global isometry or under any conformal diffeomorphism. 
\end{itemize}
In all these cases the fibres of the net of causal loops are \textbf{non-Abelian} $\rC^*$-algebras. 
\end{theorem}
\begin{proof}
We have to prove only that the fibres are non-Abelian $\rC^*$-algebras. According to the Definition (\ref{Bd:1}) 
and to the definition  of the functor $\widehat{\rL}_*$ given in Corollary \ref{Bc:8},   
it is enough to prove that \smash{$\widehat{\rL}(K|o)$} is a non-Abelian group. To this end recall 
that   \smash{$\widehat{\rL}(K|o)=\rL(K|o)/ \rC\rL(K|o)$}. We now make use of the properties of 
the posets  associated to the above spacetimes, outlined at the end of the Section  \ref{Aaa}. 
By the properties 1 and  2 follows that $\rL(K|o)$ is a free group with infinite generators. Since   
the subgroup $\rC\rL(K|o)$ of causal commutators $\rL(K|o)$
is a proper subgroup of the commutator subgroup of $\rL(K|o)$,  the quotient  $\widehat{\rL}(K|o)$
is a non-Abelian group. 
\end{proof}

\section{Representations of the net of causal loops}  
\label{C}

We have seen that a causal  net 
of $\rC^*$-algebras, the net of loops, is associated to any space(time) manifold;  this is   covariant with respect to any symmetry group of the space(time).  In this section we shall  deal with  the representations of the  net of loops. On the one hand, 
we want to understand  
how the combinatorial structure which underlies  the net of loops reflects on its representations. 
On the other hand,  we want to  understand, in the case of a topological  symmetry group,  
whether non-trivial continuously covariant representations of the covariant net exist. \\
\indent  We shall focus  on two classes  of representations, both of them   
admitting  an interpretation in terms of cohomology of posets. \\
\indent The first class is the one of representations induced by the dynamical system $(\cF(K),S,\varphi)$, equation (\ref{Bc:9}). 
These correspond to  connection 1-cochains of the poset $K$, according to the definition given in \cite{RR},   which in addition satisfy some  causal and covariant properties. By means of this
correspondence we shall point out, at the end of this section,  a sort 
of ``physical universal property" of the net of loops: any Bosonic free quantum field on a suited class of space(time)s, 
including the Minkowski spacetime  and the circle $S^1$,  provides a representation of this net. 
Moreover, in these cases,  the  construction  implies  
the existence of non-trivial, continuously covariant representations of the net.\\   
\indent The second class is the one of representations induced by the dynamical system $(\cA(K),S,\alpha)$, equation (\ref{Bc:9}).
Any such a representation results to be associated to a causal and covariant \emph{connection system}, i.e.\  
a collection of connections where causality and covariance are properties of the system  and not of  the  single connection.
Actually, the system associated to a representation is not unique; two systems associated to the  same representation    
are related by a suitable notion of  \emph{gauge} transformation.  In this  sense,  we may establish a link with 
quantum gauge theories, interpreting  such a  representation  as a gauge field  and  any 
connection system  associated to the representation  as a gauge potential. \smallskip

From now on, being interested to the applications  to quantum field theory,  
the causal posets we consider are those associated to physical spacetimes, i.e.\  the ones listed in Section \ref{Aaa}.

\subsection{Preliminaries}
\label{Ca}

Fix a causal poset $K$ with a  symmetry group $S$. 
Then $S$ acts on $K$ by causal automorphisms, i.e. 
$s \in (K,K)$,  for any $s\in S$,
and extending by functoriality the above maps we get $S$-actions by group automorphisms
on $\rF(K)$ and $\rL(K)$ which, since causality and order are preserved, factorize 
through the $S$-actions $\widehat\rF_s$, $\widehat\rL_s$, $s \in S$, 
on $\widehat{\rF}(K)$, $\widehat{\rL}(K)$ respectively.
So we get actions on the associated $\rC^*$-algebras
\[
\alpha_s = \widehat{\rL}_{*,s} \in {\bf aut}\cA(K)
\ , \
\varphi_s = \widehat{\rF}_{*,s} \in {\bf aut}\cF(K)
\ \ , \ \
 s \in S
\ .
\]
In this way the results of the previous section are summarized as follows,
\begin{equation}
\label{Ca:1}
\xymatrix{ 
(\rF(K),S)\ar[r]^{Q^F_K} & (\widehat{\rF}(K),S) \ar@{.>}[r]^{\rC^*} & (\cF(K),S,\varphi)  \\ 
(\rL(K),S)\ar[u]_{\subseteq}\ar[r]_{Q^L_K} & (\widehat{\rL}(K),S)\ar[u]^{\imath_K} \ar@{.>}[r]_{\rC^*} & (\cA(K),S,\alpha) \ar[u]_{\rho_K}
\\
& & (\cA,\mathrm{id},\alpha)_K \ar[u]_{\subseteq} }
\end{equation}
Here, on the l.h.s.\ we have a commuting diagram of groups carrying 
an $S$-action, whilst on the r.h.s.\ we have an $S$-covariant inclusion 
of the net $(\cA,\mathrm{id},\alpha)_K$ into the $\rC^*$-dynamical system $(\cA(K),S,\alpha)$ which, 
in turn, embeds in the $\rC^*$-dynamical system $(\cF(K),S,\varphi)$ defined by
$(\widehat{\rF}(K),S)$.\smallskip 

Now, a \emph{covariant representation} of the net $(\cA,\mathrm{id},\alpha)_K$ is a pair 
$(\pi,\Gamma)$ where  $\pi=\{\pi_o:\cA_o\to\cB\cH\,|\, o\in K\}$ is a family 
of representations on a fixed Hilbert space $\cH$ satisfying 
\begin{equation}
\label{Ca:2a}
\pi_a\restriction \cA_o=\pi_o \ , \qquad o\leq a \ ,
\end{equation}
and $\Gamma$ is a  unitary representation of the symmetry group $S$ on $\cH$ such that 
\begin{equation}
\label{Ca:2b}
\mathrm{Ad}_{\Gamma_L}\circ \pi_o=\pi_{Lo} \ , \qquad o\in K \ , L\in S \ . 
\end{equation}
If $\Gamma$ is strongly continuous, then we say that $(\pi,\Gamma)$ is \emph{continuously covariant}.\smallskip

The net  of causal loops $(\cA,\mathrm{id},\alpha)_K$ has non-trivial covariant 
representations  since it is embedded in $(\cA(K),S,\alpha)$, and any 
covariant representation of this dynamical system  induces a covariant 
representation of the net. Notice however that 
$(\cF(K),S,\varphi)$ induces, \emph{via} the  morphism  $\rho_K$,  
representations of the net as well.  In what follows, we shall analyze
these two classes of representations in terms of the cohomology of posets and, as a by-product,
we shall show  at the end the existence of continuously covariant representations. 

\begin{remark}
\label{Ca:3}
A net of $\rC^*$-algebras
admits more general representations when  defined over a nonsimply-connected poset \cite{RV}.  
The form of the representations we are considering in the present paper 
are called \emph{Hilbert space representations} in \cite{RV}. We restrict to this set of representations 
because any representation of the net of causal loops induced by the above dynamical systems 
has this form. 
\end{remark}

\subsubsection{A cohomological description of connections over posets} 
\label{Caa}

The notion of a connection over 
a poset was introduced in the context of 
%
%
(non-Abelian) cohomology of a poset, 
with the aim of finding a suitable framework 
for studying gauge theories in the setting of algebraic quantum field theory \cite{RR}. \\
\indent A \emph{connection} 1-cochain of a poset $K$ 
is a field $u:\Si_1(K)\ni b\to u(b)\in \cU(\cH)$ of unitary operators of a Hilbert space $\cH$ satisfying  
$u(b)^{-1}= u(\overline{b})$, $\forall b\in\Si_1(K)$, 
and the 1-cocyle equation on the nerve of $K$:
\begin{equation}
\label{Caa:1}
 u(\partial_0c) \, u(\partial_2c)= u(\partial_1c) \ , \qquad c\in\rN_2(K) \ .
\end{equation}
A connection  is \emph{flat} if it is a 1-cocycle, that is,  if the 1-cocycle equation is verified 
on all $\Si_2(K)$. This definition finds its motivation in several properties that such a notion 
share with the corresponding one  on fibred bundles over manifolds. The curvature of a
connection,  its 2-coboundary, satisfies an equation analogous of the Bianchi identity; 
a connection is flat if, and only if, its curvature is trivial; any flat connection 
defines a representation of the fundamental group of the poset $K$; there is 
an anologue of the Ambrose-Singer theorem.  
Furthermore, any connection identifies 
the fibred bundle where it lives: precisely, a (unique) 1-cocycle is associated with 
any connection, and the former is  nothing but that 
a fibred bundle over the poset $K$. 
So, a gauge transformation of a connection 
turns out to be an automorphism of the associated flat connection. 
Finally, any fibred bundle over $K$ can be interpreted as a flat  bundle 
in the sense of differential geometry when $K$ is a base of a manifold, see \cite{RRV}.

\subsection{Connections}
\label{Cb}

We now show that representations of the net of causal loops induced by the dynamical system $(\cF(K),S,\varphi)$ 
correspond to a class of connections of $K$.
To start, we generalize the notion of connection 1-cochain of a poset, recalled above, 
taking  into account the properties of causality and covariance. 
\begin{definition}
\label{Cb:1}
Given a causal poset $K$ with a symmetry group $S$. A \textbf{causal and covariant connection over $K$} 
is a pair $(u,\Gamma)$  where $u$ is  a field 
$\rT_1(K)\ni b\to u(b)\in \cU(\cH)$ of unitary operators of a Hilbert space $\cH$ and  
$\Gamma$ is a unitary representation of $S$ in $\cH$ satisfying the relations:
\begin{itemize}
\item[(1)] $ u(\overline{b})= u(b)^{-1}$ for any $b\in\rT_1(K)$; 
\item[(2)] $\Gamma_s  u(b)=  u(s(b))\, \Gamma_s$ for any $s\in S$ and $b\in\rT_1(K)$;
\item[(3)] $ u(w_1)  u(w_2)=  u(w_2) u(w_1)$  for any  $w_1,w_2\in \rL(K)$ such that $w_1\perp w_2$.
\end{itemize}
If $\Gamma$ is strongly continuous, then we say 
that   $( u,\Gamma)$ is \textbf{continuously covariant}. 
\end{definition}
\noindent Note that although $ u$ is defined only on $\rT_1(K)$ it admits an obvious extension 
$ u'$ on all $\Si_1(K)$,
\[
  u' (b):=\left\{
 \begin{array}{ll}
   u(b) \ , & b\in\rT_1(K) \ , \\
  1 \ , & b\in\rN_1(K) \ ,
 \end{array}\right.
\] 
since $\rT_1(K)$ and $\rN_1(K)$ are disjoint subsets of $\Si_1(K)$;
so $ u'$ is a connection in the sense of \cite{RR} (see Subsection \ref{Caa} above) and the results 
of that paper applies. \smallskip
%
%

\indent The relation between causal and covariant connections of $K$ and covariant representations 
of the net $(\cA,\mathrm{id},\alpha)_K$ induced by $(\cF(K),S,\varphi)$  is stated in the next  
\begin{lemma}
\label{Cb:2}
Let $K$ be a causal poset with a symmetry group $S$. Then 
there exists a bijective correspondence between covariant representations of the $\rC^*$-dynamical 
system $(\cF(K),S,\varphi)$ and causal and covariant connections of $K$.  
\end{lemma}
\begin{proof}
$(\Rightarrow)$ Let $(\pi,\Gamma)$ be a covariant representation of $(\cF(K),S,\varphi)$ and 
$(\pi^*,\Gamma)$ the induced representation of $(\widehat{\rF}(K),S)$ (see Remark \ref{Ab:4}). 
Finally let $(\check{\pi}^*,\Gamma)$ be the covariant representation 
of $(\rF(K),S)$ defined by $\check{\pi}^*:=\pi^*\circ Q^F$. By restricting 
$\check{\pi}^*$ to $\rT_1(K)$ one can easily see that the pair $(\check{\pi}^*,\Gamma)$ is 
a causal and covariant connection of $K$.\\
\indent $(\Leftarrow)$ Let $( u,\Gamma)$ be a causal and covariant connection of $K$.
Since $\rT_1(K)$ generates $\rF(K)$ we can extend $ u$ over all $\rF(K)$,
\[
 u(w) :=  u(b_n) \cdots  u(b_1)
\ \ , \ \
w = b_n \cdots b_1 \in \rF(K)
\ , \
b_k \in \rT_1(K)
\ , \
k = 1 , \ldots , n
\ .
\]
This yields a covariant representation $( u,\Gamma)$ of $(\rF(K),S)$. 
By property (3) of  Definition \ref{Cb:1}, $ u$ annihilates on the subgroup of causal commutators of $\rF(K)$,
so it admits an extension to a representation $(\hat u,\Gamma)$ 
of $(\widehat{\rF}(K),S)$. Finally, this induces 
a covariant representation $(\hat u_*,\Gamma)$ 
of $(\cF(K),S,\varphi)$, see Remark \ref{Ab:4}. \\
\indent It is easily seen that these two constructions are each other's inverses.
\end{proof}
\noindent Some observations are in order. First, this lemma  enhances the results obtained in \cite{RR}, 
since it proves the existence of non-trivial causal and covariant connections. 
Secondly, the above correspondence preserves the strongly  continuity 
of the representation of the symmetry group. So we have a 1-1 correspondence between continuously covariant 
representations of the net of causal loops induced by $(\cF(K),S,\varphi)$ and causal 
and continuously covariant connections of $K$.

%
%
%
%
\subsection{Connection systems} 
\label{Cc}
The representations of the net of causal loops induced by $(\cA(K),S,\alpha)$ generalize 
those induced by   $(\cF(K),S,\varphi)$.   But, as  in general there does not exist a 
canonical extension of covariant representations between an inclusion of $\rC^*$-dynamical systems, 
we cannot follow the route of the  previous section 
to get a correspondence with the cohomology of posets.  So, following  the scheme (\ref{Ca:1}),   
we first associate a representation of 
the group of loops $\rL(K)$ to any representation of $(\cA(K), S,\alpha)$ and, afterwards, 
associate 1-cochains of $K$ to  representations of $\rL(K)$.  \smallskip

To begin, we give the following 
\begin{definition}
\label{Cc:1}
Let $K$ be a causal poset with a symmetry group $S$. A 
\textbf{causal and  covariant  representation} of $\rL(K)$
is  a pair  $(w,\Gamma)$, where $w$ is a unitary representation
of $\rL(K)$ on a Hilbert space $\cH$ and  $\Gamma$ is a
unitary representation of $S$ on $\cH$ satisfying the relations   
\begin{enumerate}
\item  $\Gamma_s \, w(p)=w(s(p))\, \Gamma_s$ for any $s\in S$ and loop $p$;
\item  $[w(p),w(q)]=0$ for any pair $p,q$  of causally disjoint loops.
\end{enumerate}
We say 
that  $(w,\Gamma)$ is \textbf{equivalent} to another such a representation  $(w',\Gamma')$ on a Hilbert space $\cH'$ 
if there is a field $t:\Si_0(K)\ni o\to t(o)\in\cU(\cH,\cH')$ of unitary operators such that 
\[
t(o)\, w(p)=w'(p)\, t(o) \ \ , \ \  t(s(o))  \, \Gamma_s= \Gamma'_s\, t(o)  \ , \qquad   p:o\to o  \ , \ s\in S \ .   
\] 
If $\Gamma$ is strongly continuous, then we say that  $(w,\Gamma)$ is \textbf{continuously covariant}.
\end{definition}

\noindent Now, following the same reasoning of the proof of Lemma \ref{Cb:2}, it follows that: \\
\emph{(continuously) covariant representations of the 
dynamical system $(\cA(K), S,\alpha)$ are in a 1-1 correspondence  with causal and (continuously) covariant
representations of $\rL(K)$}. 
\smallskip

\indent What remains to be understood is  how to relate  such 
representations of $\rL(K)$  to  connection 1-cochains of $K$.   To this end, we recall   
a standard  procedure  to get 1-cochains from functions defined over loops.
This procedure, borrowed from algebraic topology, has been used to establish a correspondence 
between representations of the fundamental group of $K$ and 1-cocycles in   
\cite{Ruz}. Basically it makes use of  the notion of a   path-frame introduced in \cite{BR}.
A \emph{path-frame $P_o$ over a pole $o$} is a choice of a path 
$p_{(o,a)}:o\to a$ for any 
element $a$ of the poset $K$,  
subjected to the condition that $p_{(o,o)}=1$. 
Clearly, since any  causal  poset 
is pathwise connected, path frames formed by elements of $\rT_1(K)$ exist  
\footnote{This holds true for the causal posets we are considering, namely those of Section \ref{Aaa}. 
According to the definition of $\rT_1(K)$, it is enough to observe that if $b=(a;a,o)\in\rN_1(K)$, i.e.\ is  a 1-simplex  of the nerve,  
we can always find $\tilde a\Supset a$ (see Section \ref{Aaa}). 
Then the 1-simplex $(\tilde a;a,o)\in\rT_1(K)$ joins the faces 
$a$ and $o$.}.  \\
\indent So,  following \cite{BR}, given a path frame $P_o$ 
and a causal and covariant representation $(w,\Gamma)$ of $\rL(K)$ define  
\[
u(b) :=  w(p_{(o,\partial_0b)}\,b\, \overline{p_{(o,\partial_1b)}}) \ , \qquad b\in\rT_1(K) \ . 
\]
This is a connection 1-cochain since   $u(\overline b)=u(b)^*$
but, in general, $u$  is neither covariant nor causal in the sense of the Definition \ref{Cb:1}. \smallskip 

To get causality and covariance we need a \emph{covariant path-frame system}, that is  
a collection of path frames $\cP:=\{P_o\, |\, o\in K\}$ satisfying 
\begin{equation}
\label{Cc:2}
 s(P_a)= P_{s(a)} \ , \qquad a\in K \ , \ s\in S \ ,   
\end{equation}
which amounts to saying that $s(p_{(a,x)})=p_{(s(a),s(x))}$, where $p_{(a,x)}\in P_a$ whilst 
$p_{(s(a),s(x))}\in  P_{s(a)} $. Note that, 
according to this definition, if $s$ is in the stabilizer $S_a$ of $a$, then 
$s(P_a)=P_a$.\smallskip 

We now give a necessary and sufficient condition for the existence of covariant path-frame systems, 
afterwards we apply this result to two remarkable examples. 
\begin{proposition}
\label{Cc:3}
Let $K$ be a causal poset with a symmetry group $S$. 
$K$ has a covariant path-frame system if, and only if, 
for any pair $o,a\in K$  there is a path joining $a$ to $o$ 
whose generators  are  invariant under   $S_a\cap S_o$. 
\end{proposition}
\begin{proof}
($\Rightarrow$) Let $\cP$ be a covariant path-frame system. 
According to (\ref{Cc:2}) if $s$ is an element of the  stabilizer $S_a$ of $a$, then 
$s(P_a)=P_a$. Let $y$ be such that $S_y\cap S_a\ne 1$. Then for any 
$s\in S_y\cap S_a$ we have that $s(p_{(a,y)})= p_{(a,y)}$. This, according to (\ref{Bb:5}) and (\ref{Bb:6}), amounts to saying that 
all the 1-simplices generating the path $p_{(a,y)}$ are  invariant under  $S_a\cap S_y$.\\
\indent ($\Leftarrow$) For any element  in the orbit space $K/S$ choose a representative  $a$. 
We first construct  a suited path frame over $a$. To this end consider  the orbit space $K/S_a$.   
For any orbit in this space choose a representative $y$  and define $p_{(a,y)}$ as  
a path whose generating 1-simplices are invariant under $S_a\cap S_y$.  Clearly if $S_a\cap S_y=1$ no  restriction 
is imposed on the choice of the path.  Afterwards if $z= sy$ for some $s\in S_a$ define 
\[
 p_{(a,z)}:= s(p_{(a,y)})
\]
This defines a path frame $P_a$ over $a$. Note that if  $r\in S_a$, then  
\begin{equation}
\label{eq:pathframe}
r(p_{(a,z)})= rs(p_{(a,y)}) = p_{(a,r(z))} \ .
\end{equation}
So $P_a$ agrees with the above observation.  We now construct a path frame $P_{a'}$ for any element $a'$ in the orbit of $a$. 
If $a'=s(a)$  for $s\in S$, then  we set 
\[
 p_{(a',y)}:= s(p_{(a,s^{-1}(y))}) \ , \qquad y\in K \ . 
\]
This definition is well posed. In fact if $r(a)=a'$ then $s^{-1}r$ belongs to the stabilizer of $a$. Hence by 
equation (\ref{eq:pathframe}) we have 
\[
 r(p_{(a,r^{-1}(y))}) = s (s^{-1}r)(p_{(a,r^{-1}y)}) = s (p_{(a,(s^{-1}rr^{-1})(y))})=  s (p_{(a,s^{-1}y)}) \ . 
\]
The collection $P_a$, with $a\in K$, is manifestly covariant.  
\end{proof}
\noindent On this  ground, the cases of Minkowski spacetime and $S^1$ can be completely understood. 
\begin{corollary}
\label{Cc:4}
The following assertions hold.
\begin{itemize}
\item[(i)] The set of double cones of the  Minkowski spacetime  
has a covariant path-frame system with respect to the Poincar\'e group. 
\item[(ii)] The set of nonempty, open, connected intervals of $S^1$, having a proper closure, 
does not have a covariant path-frame system with respect to the M\"obius group. 
\end{itemize}
\end{corollary}
\begin{proof}
$(i)$ Fix a reference frame of the Minkowski spacetime. According to the definition of a double cone given in Section \ref{Aaa}, 
it is enough to  study the following situation: $o_R$ is a double cone based on the subspace $C_0$ at time $t=0$ 
whose base is centered at the origin of the reference frame; $o$ is any other double cone.  
Recall that the stabilizer of $o_R$ is $SO(3)$, the subgroup of spatial rotations. 
Now, let $S=S_{o_R}\cap S_o$ be the intersection of the stabilizers  
of $o_R$ and $o$. Take $R'>0$ large enough that 
$o_R,o\subseteq o_{R'}$.  Then $S= S_{o_R}\cap S_o\cap S_{o_{R'}}$ because $o_R$ and $o_{R'}$ have the same stabilizer. Hence the 1-simplex $(o_{R'}; o,o_R)$  satisfies the condition of  Proposition \ref{Cc:3}.\\
\indent $(ii)$ We show that 
the condition  of   Proposition \ref{Cc:3} is not verified. As observed in Section \ref{Aaa},
the  stabilizer of an interval $a$, which is isomorphic to the modular group,  is also 
the stabilizer of the open interval  $S^1\setminus \{cl(a)\}$ of $a$. Since the modular group acts 
transitively within these intervals,  there is no other interval having the same stabilizer. 
\end{proof}
\noindent Just a comment about this result. Actually the above proofs seem to point out 
a topological obstruction to the existence of a  path-frame system. The reason why 
the poset of double cones of the  Minkowski spacetime has path-frame systems 
relies on the fact that this poset is upward directed. This implies that the poset is contractible, 
see \cite{Ruz}. Clearly this does not happen for the intervals of the circle, whose associated
poset has homotopy group $\bZ$.\smallskip

Given a  covariant path-frame system  $\cP$, we associate with any $o\in K$ the connection 1-cochain 
\begin{equation}
\label{Cc:5}
\ru_{\cP_o}(b) :=  w(p_{(o,\partial_0b)}\,b\, \overline{p_{(o,\partial_1b)}}) \ , \qquad b\in\rT_1(K) \ . 
\end{equation}
Clearly, as before, if we look at a single connection $\ru_{\cP_o}$ we have neither covariance nor
causality, however these two properties arise if we look at  the collection $\ru_{\cP}:=\{\ru_{\cP_o} \ ,  \ o\in K\}$. In fact 
since we are using a covariant path-frame system   we have 
\[
\mathrm{Ad}_{\Gamma_s}(\ru_{\cP_o}(b)) =   
\mathrm{Ad}_{\Gamma_s}(w(p_{(o,\partial_0b)}\,b\, \overline{p_{o,\partial_1b}})) = 
 w(s(p_{(o,\partial_0b)})\,s(b)\, \overline{s(p_{o,\partial_1b})}) = 
 \ru_{\cP_{s(o)}}(s(b)) \ . 
\] 
Moreover,  $w$ is  recovered from the system $\ru_\cP$  
by observing that $w(p) = \ru_{\cP_a}(p)$, for any loop $p:a\to a$, because 
$p_{(a,a)}=1$.
Causality follows from this relation, in the sense that $[\ru_{\cP_a}(p),\ru_{\cP_o}(q)]=0$
for any $p:a\to a$, $q:o\to o$ with $p\perp q$.\smallskip 

Hence we may give the following  
\begin{definition}
\label{Cc:6}
Let $K$ be a causal poset with a symmetry group $S$.  
A \textbf{causal and covariant connection system} is a 
pair $(\ru,\Gamma)$ where  $\ru$ is a family $\{\ru_o \ , \ o\in K\}$ of  connection 1-cochains of $K$ in a Hilbert space $\cH$
and $\Gamma$ is a unitary representation of $S$ on $\cH$ such that 
\begin{enumerate}
\item $[\ru_o(p),\ru_a(q)]=0$, if $p\perp q$ and $p:o\to o$, $q:a\to a$. 
\item $\mathrm{Ad}_{\Gamma_s} \circ \ru_o= \ru_{s(o)}\circ s$,  for any  $s\in S$ and $o\in K$.
\end{enumerate}
We say that $(\ru,\Gamma)$ is \textbf{equivalent} to $(\ru',\Gamma')$ if there is a family $\rt:=\{\rt_o \ , \ o\in K\}$ 
of fields  $\rt_o:\Si_0(K)\ni a\to \rt_o(a)\in\cU(\cH,\cH')$ of unitary operators  satisfying the relations  
\[
\rt_o(\partial_0b) \ru_o(b)= \ru'_o(b) \rt_o(\partial_1b)  \ \ , \ \ 
\rt_{s(o)}\circ s \, = \Gamma'_s \rt_o \Gamma_{s^{-1}} \ , \ \ \ \  b\in\rT_1(K) \ ,  \ s\in S \ . 
\]
If $\Gamma$ is strongly continuous, then we say that $(\ru,\Gamma)$ is \textbf{continuously covariant}.
\end{definition}
\noindent Note that if $(u,\Gamma)$ is a connection as in Definition \ref{Cb:1}
then setting $\ru_o:=u$ for any $o$ yields a causal and covariant connection system $(\ru,\Gamma)$.
Thus the notion of a connection system generalizes that of a connection.
\smallskip   

Now, it is easy to see that  any  causal and covariant  connection system $(\ru,\Gamma)$ defines a 
causal representation of the dynamical system $(\cA(K),S,\alpha)$. To this end, define 
\begin{equation}
\label{Cc:7}
 w(p):= \ru_o(p) \ , \qquad p:o\to o \ .  
\end{equation}
Clearly $w(p)^{*}=w(\overline{p})$ for any loop $p$; causality and covariance of $w$ follow from properties 1. and 2. of 
Definition \ref{Cc:6}.  
In particular, concerning covariance, if $p:o\to o$ then    
\[
 \mathrm{Ad}_{\Gamma_s}(w(p)) = \mathrm{Ad}_{\Gamma_s}(\ru_o(p)) = 
 \ru_{s(o)}(s(p)) =w(s(p)) \ . 
\]
The observation below Definition \ref{Cc:1} and the equations (\ref{Cc:5})  and  (\ref{Cc:7}) establish a correspondence between causal and covariant connection systems of $K$ and 
representations of the net of causal loops induced by $(\cA(K),S,\alpha)$. 
One can easily check that this correspondence is, up to equivalence, bijective.  
\begin{proposition}
\label{Cc:8} 
There is, up to equivalence, a 1-1 correspondence between causal and (continuously) covariant  connection systems 
of $K$ and (continuously) covariant representations of the net of causal loops induced by   $(\cA(K),S,\alpha)$. 
\end{proposition}
%


\subsection{Gauge transformations}
\label{Cd}

Intuitively a gauge transformation of a causal and covariant connection system 
is a transformation preserving  covariance and  causality. More precisely,
\begin{definition}
\label{Cd:1}
A \textbf{gauge transformation} 
of a causal and covariant  connection system  $(\ru, \Gamma)$ is  a collection  
$\rg:=\{\rg_a \ , \ a\in K\}$ of fields  $\rg_a:\Si_0(K)\ni o\to \rg_a(o)\in\cU\cH$  of unitary operators  
satisfying the following properties
\begin{enumerate}
\item  $\rg_{s(a)}(s(o))\, \Gamma_s = \Gamma_s\, \rg_a(o)$, for any $a,o\in K$ and any $s\in S$; 
\item $\mathrm{Ad}_{\rg_a(a)}(\ru(\cA_o))= \ru(\cA_o)$ for any inclusion  $a\leq o$, 
\end{enumerate}
where by $\ru(\cA_o)$ we mean  the fibre  
$\cA_o$ of the net of causal loops in the representation induced  by $\ru$. We denote 
the set of gauge transformations by $\cG$.
\end{definition}
We now draw on some consequences of this definition.  \emph{First}, given a gauge transformation $\rg$ of $(\ru,\Gamma)$,
we define 
\begin{equation}
\label{Cd:1a}
 \ru^\rg_o(b):= \rg_o(\partial_0b)\, \ru_o(b) \, \rg_o^*(\partial_1b) \ , \qquad b\in\rT_1(K) \ , 
\end{equation}
for any $o\in K$, and observe that  the pair  $(\ru^\rg,\Gamma)$ is a causal and covariant connection system. 
Covariance easily follows from the  definition of gauge transformation.  Concerning causality,  
let $p:o\to o$ and $q:a\to a$ be two causally disjoint loops. This, according to (\ref{Bb:4}), amounts to saying 
that there are $\tilde o,\tilde a\in K$ such that  $|p|\leq \tilde o$ and $|q|\leq \tilde a$ and $\tilde o\perp\tilde a$. 
Property 2. of Definition \ref{Cd:1}  implies that 
\[
\ru^\rg_o(p)\in \ru(\cA_{\tilde o}) \  \ , \  \ \ru^\rg_a(q)\in \ru(\cA_{\tilde a}) \ . 
\]
Hence  $\ru^\rg_o(p)$ and $\ru^\rg_a(q)$ commute because so do the algebras $\cA_{\tilde o}$ and $\cA_{\tilde a}$.  
\emph{Secondly}, by (\ref{Cd:1}) the causal and covariant representation    $(w,\Gamma)$ of $\rL(K)$ associated with $(\ru,\Gamma)$ by (\ref{Cc:7}),  transforms, under a gauge transformation $\rg$, as 
\begin{equation}
\label{Cd:3}
 w^\rg(p)= \rg_o(o)\, w(p)\,\rg^*_o(o) \ , \qquad p:o\to o \ .  
\end{equation} 
\emph{Thirdly},  as one can easily deduce by Definition \ref{Cd:1}, gauge transformations form a group under the composition 
\[
(\rg\cdot \rh)_a(o):= \rg_a(o)\, \rh_a(o) \ , \qquad a,o\in K \ .  
\]
%
%
%
So we shall refer to $\cG$ as the group of \emph{gauge transformations}. \emph{Fourthly}, for any $o\in K$, let $G_o$ the subgroup of $\cU\cH$ generated by  
\begin{equation}
\label{Cd:4}
\{ \rg_a(a) \in\cU\cH \, | \, \rg\in\cG \ , \ a \leq o \}   \ .  
\end{equation}
The group $G_o$ acts by automorphisms upon the algebra $\ru(\cA_o)$. We refer to   $G_o$  as  \emph{the  gauge group  over $o$.}
Note that according to Definition \ref{Cd:1},  we have that $G_o\subseteq G_a$ for any inclusion  $o\leq a$ 
and that  $G_o$ is spatially isomorphic to $G_{s(o)}$ for any $s\in S$. So we have a \emph{covariant net of gauge groups}  
$(G,\mathrm{id},\mathrm{Ad}_\Gamma)_K$ whose fibres are the gauge groups $G_o$.\smallskip

As an interesting case, let us analyze how  the connection system associated with a causal and covariant representation $(w,\Gamma)$  
of $\rL(K)$  depends on the choice of the  path-frame system. Consider two path-frame systems  
$\cP$ and $\cQ$ and the corresponding connection systems $\ru_\cP$ and $\ru_\cQ$. To be precise 
 by equation (\ref{Cc:5}),   we have 
\[
\ru_{\cP_o}(b)=w(p_{o,\partial_0b}\,b\, \overline{p_{(o,\partial_1b)}}) \ \ , \ \ 
\ru_{\cQ_o}(b)=w(q_{o,\partial_0b}\,b\, \overline{q_{(o,\partial_1b)}})  \ \ , \ \ 
\]
for any $b\in\rT_1(K)$ and $o\in K$. Then observe that if we set  
\[
\rg_o(a):= w(q_{(o,a)} \overline{p_{(o,a)}})  \ , \qquad a\in K \ , 
\] 
since $g_o(o)=1$ for any $o$, then we get a gauge transformation $\rg$ of $(\ru_\cP,\Gamma)$ such that 
$\ru^\rg_\cP=\ru_{\cQ}$. So a changing of  a path-frame system leads to a gauge transformation of the connection system.  
Moreover,  $w$ is invariant 
under this gauge transformation, that is $w^\rg=w$. 

\subsection{Existence of continuously covariant representations} 
\label{Ce}

In this section we focus on the Minkowski spacetime $\mathbb{M}^4$ and prove  
the existence  of continuously Poincar\'e-covariant  (in symbols $\mathscr{P}^\uparrow_+$-covariant)  
representations of the net of causal loops over the set of double cones of $\mathbb{M}^4$. 
We also point out that this result applies to $S^1$, when
the symmetry is given by the M\"obius group, up to some modifications . \\ 
\indent The approach  we shall use  points out the  universality 
of the net of causal loops.  In fact,  we shall see that  any Hermitian scalar quantum field  over  $\mathbb{M}^4$ 
gives a causal and continuously  \smash{$\mathscr{P}^\uparrow_+$}-covariant connection, hence a   continuously  $\mathscr{P}^\uparrow_+$-covariant 
representation of  net of causal loops.  
We deserve a further investigation to analyze the explicit physical interest of this correspondence. \\
\indent Concretely, we first define  a  \emph{translator} that allows us to pass from 
quantum fields to connections. This is a 1-cochain of the poset taking values in the set of real,  compactly 
supported  smooth functions of the spacetime, which is invariant  under the action of the symmetry group.

\subsubsection{From simplices to test functions}
\label{Cea}
In this section we define  the above cited translator and prove its existence in the case of the posets associated 
with the  Minkowski spacetime and with $S^1$. \smallskip

Consider   a spacetime $X$ and  denote 
the associated  causal  disjointness relation by   $\perp$. 
Let $S$ be a symmetry group of $X$, i.e. a  subgroup of  conformal diffeomorphisms  of $X$.  
Concretely, if $X$ is the Minkowski spacetime then $S$ is the Poincar\'e group; if $X$ is the circle 
then $S$ is the M\"obius group.  Take
a base $K$ of $X$  of open and relatively compact subsets of $X$, such that $K$ is globally invariant under the action of $S$.

\indent We denote the space of compactly supported,  real valued, 
smooth functions on $X$ by $\cD(X)$.  Similarly to the definition of net of $\rC^*$-algebras  given in Section \ref{Ac}, 
we have a  net of vector spaces
$(\cD,\mathrm{id})_K$ assigning to each element $a$ of $K$ the space $\cD_a$ of those functions of $\cD(X)$ vanishing outside 
the closure $cl(a)$ in $X$ of $a$.
For any $n \in \bN$, we consider the \emph{$n$-cochain} vector space
\begin{equation}
\label{Cea:1}
\rC^n(K,\cD) := \{  f : \Si_n(K) \ni x \mapsto f_x \in \cD_{|x|}  \} \ .
\end{equation}
There is a coboundary operator $\rd:\rC^n(K,\cD)\to \rC^{n+1}(K,\cD)$ defined, as usual, as 
\begin{equation}
(\rd f)_x=\sum^{n}_{k=0} (-1)^k f_{\partial_k x} \ , \qquad x\in \Si_{n+1}(K) \ , \ f\in \rC^n(K,\cD) \ , 
\end{equation}
satisfying the equation $\rd\circ\rd=0$. This allows us to define cohomology groups, but  this is  out of the aims 
of the present paper.  There is a left action of the symmetry group $S$  on $\rC^n(K,\cD)$ defined by 
\begin{equation}
\label{Cea:2}
(sf)_x := f_{s(x)}\circ s \ , \qquad x\in\Si_n(K) \ ,  \ s\in S \ . 
\end{equation}
Now, considering the fixed-point subspace of $\rC^n(K,\cD)$ under the action of $S$, i.e. 
\begin{equation}
\label{Cea:3}
\rC^n(K,\cD)^S:= \{ f\in \rC^n(K,\cD) \, | \,  (sf)_x= f_{x} \ , \   x\in\Si_n(K) \ ,  \ s\in S \} \ ,
\end{equation}
we  note that 
\begin{equation}
\label{Cea:4}
 \rC^n(K,\cD)^S\ni f \ \ \iff  \ \ f_{s(x)}  = f_x \circ s^{-1} \ , 
\end{equation}
for any $x\in\rT_n(K)$  and  $s\in S$. We also observe that
if $f\in\rC^n(K,\cD)^S$, then   $\rd f\in \rC^{n+1}(K,\cD)^S$
because the  coboundary operator $\rd$ commutes 
with the action of the symmetry group. In this regards one should note that $\rd f$ might be equal to $0$ even if 
$f$ is not. This happens when $f$ is an $n$-cocycle.\\
\indent For our aims it is of importance to understand 
when $\rC^n(K,\cD)^S$ is not trivial, i.e.\ different from $0$.  The next result clarifies this point. 
\begin{proposition} 
\label{Cea:5}
 The space $\rC^n(K,\cD)^S$ is not trivial if, and only if,   there is $y\in \Si_n(K)$ and a nonzero 
function $f_0 \in \cD_{|y|}$  such that $f_0 \circ s^{-1}=f_0$ 
for any $s$ in the stabilizer $S_{y}$ of $y$. 
\end{proposition}
\begin{proof}
($\Rightarrow$) Let $f\in\rC^n(K,\cD)^S$ be a nonzero function. So $f_y\ne 0$ for some $n$-simplex $y$. 
By (\ref{Cea:4}),  $f_{y}=f_y\circ s^{-1}$ for any  $s\in S_y$.\\  
\indent ($\Leftarrow$) 
Let $f_0$ and $y$ be as above. For any $x\in\Si_n(K)$,  define 
\[
 f_x:= \left\{
 \begin{array}{cc}
 f_0 \circ s^{-1} \ , &  if \ x=s(y) \ , \\  
 0 &  otherwise \ . 
 \end{array}\right.
\]
Note that $f_y=f_0$ and that $f$  is different from zero only on  the orbit of $y$. 
This definition is well posed because if $x=r(y)$ then $s^{-1}r \in S_{y}$,
so that $f_0\circ s^{-1} = f_0\circ (s^{-1}r)r^{-1} = f_0\circ r^{-1} $. 
Thus for any $x$ in the orbit of $y$, according to the definition of the action of $S$, 
we have
\[
 (s'f)_{x} = f_{s'(x)}\circ s' = f_{s's(y)} \circ s' = f_0\circ (s's)^{-1} s'  =f_0\circ s^{-1}= f_x \ , 
\]
for any $s' \in S$, completing the proof. 
\end{proof}
\begin{remark}
\label{Cea:6} 
Using the same idea of the above proof we can prove the following stronger property: 
there exists  $f\in\rC^n(K,\cD)^S$ which does not vanish on \emph{any} $n$-simplex $y$ for which there 
is a nonzero function $f_0\in\cD_{|y|}$ satisfying $f_0\circ s^{-1} = f_0$ for any $s\in S_y$. 
So, in particular, $f$ does not vanish on the $n$-simplices over which $S$ acts freely. 
\end{remark}
We shall use this result later to show the existence of invariant cochains for the Minkowski spacetime and for the circle. 
We now make a step toward the main result of this section  and specialize to the set of $1$-cochains.  
Note that $\rC^1(K,\cD)$ is a $\bZ_2$-graded vector space  under the action
$\rC^1(K,\cD)\ni f \to \overline{f}\in \rC^1(K,\cD)$ defined as 
\begin{equation}
\label{Cea:7}
\overline f_b := f_{\overline b}
\  , \qquad   b \in \Si_1(K)
\ .
\end{equation}
This yields the direct sum decomposition in the even/odd spectral subspaces
\begin{equation}
\label{Cea:8}
\rC^1(K,\cD) \to\rC^{1,ev}(K,\cD) \oplus \rC^{1,odd}(K,\cD)
\ \ , \ \
f \mapsto f^{ev} \oplus f^{odd}
\ ,
\end{equation}
where 
\begin{equation}
\label{Cea:9}
f^{ev} :=   \frac{1}{2} (f + \overline{f}) 
\ \ , \ \
f^{odd} :=  \frac{1}{2} (f - \overline{f})
\end{equation}
are such that 
\begin{equation}
\label{Cea:10}
\overline {f^{ev}}  =   f^{ev}
\ \ , \ \
\overline {f^{odd}} = - f^{odd}
\ .
\end{equation}
Note that 
$(\rd f)^{ev}=0$ while $(\rd f)^{odd}= \rd f$ for any $f\in\rC^0(K,\rD)$. 
Now,  since the $\bZ_2$-grading commutes with the induced $S$-action,
the direct sum
decompositions (\ref{Cea:8}) and (\ref{Cea:9}) apply also to the subspace $\rC^1(K,\cD)^S$: 
\begin{equation}
\label{Cea:11}
\rC^1(K,\cD)^S \to\rC^{1,ev}(K,\cD)^S \oplus \rC^{1,odd}(K,\cD)^S
\ \ , \ \
f \mapsto f^{ev} \oplus f^{odd}
\ .
\end{equation}
Let us consider the mapping $\delta:\rC^0(K,\cD)\to \rC^1(K,\cD)$
defined as 
\begin{equation}
\label{Cea:12}
(\delta f)_b:= f_{\partial_0b} - f_{\partial_1b} +f_{|b|}= (\rd f)_b +f_{|b|} \ \ , \ \ b\in\Si_1(K) \ ,
\end{equation}
which is twist of the coboundary operator $\rd$ at the level of $0$-cochains. Observe that $\delta$, like $\rd$,   commutes  
with the $S$-action; so  $\delta:\rC^0(K,\cD)^S\to \rC^1(K,\cD)^S$. Concerning the grading  $(\ref{Cea:8})$ we have  that 
\begin{equation}
\label{Cea:13}
(\delta f)^{ev}_b = f_{|b|} \ \ , \ \ (\delta f)^{odd}_b = (\rd f)_{b} \ , \qquad b\in\Si_1(K) \ .  
\end{equation}
We now apply these results to  posets associated with the Minkowski spacetime and to  the circle. 
\begin{proposition}
\label{Cea:14}
Let $K$ be the set of double cones of the Minkowski spacetime $\mathbb{M}^4$.  Then 
there exists 
$f\in\rC^0(K,\cD)^{\mathscr{P^\uparrow_+}}$ such that $f_a\ne 0$ and is nonnegative for any $0$-simplex $a$. 
Furthermore, the  1-cochain $\delta f \in \rC^1(K,\cD)^{\mathscr{P^\uparrow_+}}$ satisfies the following properties 
\begin{itemize}
\item[(i)] $(\delta f)^{ev}_b= f_{|b|}\ne 0$ is non-negative for any 1-simplex $b$; 
\item[(ii)]  $(\delta f)^{odd}_b= f_{\partial_0b} - f_{\partial_1b}\ne 0$ for any 1-simplex $b$ such that $\partial_0b\ne\partial_1b$. 
\end{itemize}
\end{proposition}
\begin{proof}
As observed in Section \ref{Aaa},   any double cone  lays in the orbit 
of a double cone $o_R$ whose base belongs  to the subspace  of the Minkowski spacetime at $t=0$,  is centred at  the origin of 
the reference frame and has radius equal to $R>0$. The stabilizer of any such double 
cone is the subgroup of spatial rotations $SO(3)$. 
So,  consider  a double cone $o_R$ and 
let $h:\mathbb{R}\to\mathbb{R}$ be a non-negative smooth function supported in $[-1,1]$. Define 
\begin{equation}
\label{Cea:14a}
 f_R(t,x,y,z):= h\left(2\cdot \frac{t}{R}\right) \cdot h\left(4\cdot \frac{x^2+y^2+z^2}{R^2}\right) \ .
\end{equation}
Here $f_R$ is a nonnegative  smooth function supported in $o_R$ that satisfies $f_R\circ s^{-1}=f_R$ for any 
$s\in SO(3)$. So, applying Proposition  \ref{Cea:5} and the Remark \ref{Cea:6},  $f_R$ defines   
an invariant 0-cochain $f$ such that $f_a\ne 0$ for any $a\in \Si_0(K)$. \\
\indent  As observed above \smash{$\delta f\in \rC^1(K,\cD)^{\mathscr{P^\uparrow_+}}$}. Using the equation (\ref{Cea:13})
 $(i)$   is obvious, while  $(ii)$ easily follows from the  definition of $f$. 
\end{proof}
The argument used in the proof of the previous proposition
cannot be applied to the circle with  the M\"obius group as a symmetry group.  The stabilizer of an interval 
acts transitively within the interval (see Section \ref{Aaa}). 
So there does not 
exist any invariant non-zero smooth function supported within the interval. Hence, according to Proposition \ref{Cea:5},   
there does not exist
any non-trivial,  M\"obius invariant 0-cochain. However this is not the case for 1-cochains.  
\begin{proposition}
\label{Cea:15}
Let $K$ be the set of open, connected intervals of $S^1$ having a proper closure.  Then there exists 
$f\in\rC^1(K,\cD)^{\mbox{M\"ob}}$  such that $f_b\ne 0$ and non-negative for any non-degenerate 1-simplex.
\end{proposition}
\begin{proof}
The M\"obius group acts freely on non-degenerate 1-simplices (see Section \ref{Ba} for the definition of a degenerate 1-simplex).
Therefore, we set $f_b=0$ if $b$ is a degenerate 1-simplex. For non-degenerate 1-simplices, we choose a representative 
element $y\in\Si_1(K)$ for any orbit  and take a non-zero and non-negative smooth function $f_0$ supported within $|y|$. 
Then proof follows 
from  Proposition \ref{Cea:5}  and Remark \ref{Cea:6}. 
\end{proof}

\subsubsection{From quantum fields to causal  and continuously covariant connections}
\label{Ceb}
We now prove the existence of continuously covariant representations of the net of causal loops 
in the case of the Minkowski spacetime. This will be done by showing that any scalar quantum field 
satisfying the Wightman axioms, and other two standard properties,  
induces a causal  and continuously covariant connection.  For brevity, we shall outline the features of a 
quantum field  strictly necessary for the aims of the  present paper  and refer the reader to the standard textbooks  
in axiomatic quantum field theory \cite{SW,BLT,GJ,RS,Ara,Haa}  for a complete description  of the  Wightman axioms. \smallskip

Let $(\Phi,\Gamma,\cH)$ be an Hermitian  scalar quantum field  in the Minkowski spacetime $\mathbb{M}^4$ 
satisfying the Wightman axioms.
The field $\Phi$ is an  operator valued distribution defined on a dense domain $D$ of a separable Hilbert space $\cH$: 
a linear mapping  $f\mapsto \Phi(f)$ associating an unbounded operator $\Phi(f)$  to any  
Schwartz function $f\in\cS(\mathbb{R}^4)$ such that $\Phi(f)$ and its adjoint $\Phi(f)^*$ map $D$ into $D$.
$\Gamma$  is a strongly continuous unitary representation of  the  Poincar\'e  group 
\smash{$\mathscr{P}^{\uparrow}_{+}$} on $\cH$ implementing  the covariance of the field: 
\begin{equation}
\label{Ceb:0}
\Gamma_s:D\to D \ \ , \ \ 
\mathrm{Ad}_{\Gamma_s}(\Phi(f)) D = \Phi(f\circ s^{-1}) D \ , \qquad  s\in\mathscr{P}^{\uparrow}_{+} \ , \  
f\in\cS(\mathbb{R}^4) \ . 
\end{equation}
%
%
In what follows we focus on two properties  that are a natural strengthening of the axioms and that are relevant for our aims. 
First, we assume that the field is \emph{self-adjoint},  namely 

\begin{itemize}
\item[{\bf P1.}]  $\Phi(f)$ is an essentially self-adjoint operator on the domain $D$, for any real Schwartz function $f\in\cS_\rr(\mathbb{R}^4)$.  
\end{itemize}
From now on, we shall work with the closure of  $\Phi(f)$  that, 
with a little abuse of notation, we shall denote by the same symbol. 
Using  {\bf P1}  we can define the exponential $\exp(i\Phi(f))$ for any  $f\in \cS_\rr(\mathbb{R}^4)$. It turns out that 
$\exp(i\Phi(f))$  is a unitary operator  of $\cH$  satisfying the relations    
\begin{equation}
\label{Ceb:1}
\exp(i\Phi(-f))  = \exp(i\Phi(f))^* \  \  \ \  \mbox{and}  \ \ \ \ 
\mathrm{Ad}_{\Gamma_s}( \exp(i\Phi(f)))  = \exp(i\Phi(f\circ s^{-1})) \ , 
\end{equation}
for any $f\in\cS_{\rr}(\mathbb{R}^4) $ and for any $s\in \mathscr{P}^{\uparrow}_{+}$; the latter relation, in particular, 
derives from (\ref{Ceb:0}). Secondly,  we assume the property 
of \emph{strong causality} of the field, namely
\begin{itemize} 
\item[{\bf P2.}]  For any pair $f,g\in\cS_{r}(\mathbb{R}^4) $ 
\begin{equation}
\label{Ceb:3}
supp(f)\perp supp(g) \ \ \ \Rightarrow \ \ \ \ 
             [\exp(i\Phi(f)), \exp(i\Phi(g))]= 0 \ .   
\end{equation}
\end{itemize}  
We stress that {\bf P1} and {\bf P2} are not required in the usual scheme 
of the Wightman axioms. However,  these two properties are verified by any scalar field satisfying the Osterwalder-Schrader 
axioms  \cite{GJ} and, in particular, by  any known example of a Hermitian scalar quantum field. \smallskip

We now are ready to prove the main results of this section.
\begin{theorem}
\label{Ceb:6}
Let $(\Phi,\Gamma,\cH)$ be a Hermitian scalar quantum field in the Minkowski spacetime $\mathbb{M}^4$
satisfying the Wightman axioms and the  properties ${\bf P_1}$ and ${\bf P_2}$. 
Let $K$ be the set of double cones in the Minkowski spacetime.  
Given a non-zero $f\in\rC^1(K,\cD)^{\mathscr{P}^\uparrow_+}$, define  
\begin{equation} 
\label{Ceb:7}
 [\Phi f](b):= \exp\left\{i \int_{b_\circledcirc }|f^{ev}_b(x)| \rd^4x \cdot \Phi\left(f^{odd}_b \right)\right\}  \ , \qquad b\in\Si_1(K) \ , 
\end{equation}
 where $\rd^4x$ is the Lebesgue measure of $\mathbb{R}^4$ and $b_\circledcirc:= |b| - ( \partial_1b \cup \partial_0 b)\subseteq X$ shall be called \textbf{the corona of} $b$.   
Then the pair $([\Phi f],\Gamma)$ is a causal and continuously $\mathscr{P}^\uparrow_+$-covariant  connection of $K$. In particular $[\Phi f](b)=1$ for any $b\in\rN_1(K)$. 
\end{theorem}
\begin{proof}
 Clearly $[\Phi f](b)=1$ when $b\in\rN_1(K)$, because in this case $b_\circledcirc=\emptyset$. Since 
$f^{odd}_{\overline{b}}=- f^{odd}_{b}$, $f^{ev}_{\overline{b}}=f^{ev}_{b}$  and  
$\overline{b}_\circledcirc = b_\circledcirc$, by  (\ref{Ceb:1})
we have that $[\Phi f](\overline b) = [\Phi f](b)^*$. This proves that $[\Phi f]$ is a connection 1-cochain.    Causality follows by property $\mathbf{P_2}$. About covariance, since 
\smash{$f^{ev}, f^{odd} \in\rC^1(K,\cD)^{\mathscr{P}^\uparrow_+}$},   relation (\ref{Cea:4}) implies that 
$ f^{ev}_{s(b)}=  f^{ev}_{b}\circ s^{-1}$   and   $f^{odd}_{s(b)}=  f^{odd}_b\circ s^{-1} $ for any  1-simplex $b$ and any Poincar\'e transformation  $s$.   
These relations, the equation (\ref{Ceb:1}) and  the invariance of the Lebesgue measure 
under transformations of the Poincar\'e group, imply that the pair $([\Phi f],\Gamma)$ is a causal and continuously covariant 
connection of $K$. In fact,  since $s(b_\circledcirc)= s(b)_\circledcirc$ for any 1-simplex $b$ and any Poincar\'e transformation 
$s$,   we have that 
$\int_{b_\circledcirc }|f^{ev}_b(x)| \rd^4x = \int_{s(b)_\circledcirc}|f^{ev}_b(s^{-1}(x))|  \rd^4x  = 
\int_{s(b)_\circledcirc}|f^{ev}_{s(b)}(x)|  \rd^4x$. Then     
\begin{align*}
\mathrm{Ad}_s([\Phi f](b)) & = 
\exp\left\{i \int_{b_\circledcirc }|f^{ev}_b(x)| \rd^4x \cdot \Phi\left(f^{odd}_b\circ s^{-1} \right)\right\} \\ & =   
\exp\left\{i \int_{s(b)_\circledcirc }|f^{ev}_{s(b)}(x)| \rd^4x \cdot \Phi\left(f^{odd}_{s(b)} \cdot \right)\right\} = 
 [\Phi f](s(b))  \ , 
\end{align*}
and this completes the proof. 
\end{proof}
Hence any Hermitian scalar quantum field defines a causal and continuously $\mathscr{P}^\uparrow_+$-covariant connection  
and  this  defines, in turn,  a continuously  \smash{$\mathscr{P}^\uparrow_+$}-covariant representation of the net of causal loops. 
However the general hypothesis of the previous theorem  
cannot exclude the triviality of this construction, i.e.\  $\Phi$ might annihilate on the 1-cochain $f^{odd}$ even  
if  $f^{odd}$ is not trivial like, for instance,  that constructed in Proposition \ref{Cea:14}. \\
\indent We avoid this eventuality in the case of  
the free Hermitian scalar field $(\Phi_m,\Gamma_m,\cH_m)$, of mass $m\geq 0$, in $\mathbb{M}^4$.    
This field  satisfies all the properties outlined above, and
we briefly recall its definition following closely  the reference \cite{RS}. 
The Hilbert space $\cH_m$ is the symmetric Fock  space associated to the 1-particle Hilbert space $L^2(\rH_m,\rd\Omega_m)$ where 
$\rH_m$ is the hyperboloid of mass $m\geq 0$ of $\mathbb{R}^4$ and $\rd\Omega_m$ is the  
\smash{$\mathscr{L}^\uparrow_+$}-invariant measure on $\rH_m$.
$\Gamma_m$ is a strongly continuous unitary representation of the Poincar\'e group \smash{$\mathscr{P}^\uparrow_+$} on $\cH_m$, and the generators of the translation subgroup  
have joint spectrum on the closure of the forward lightcone 
\smash{$\overline{V}_+$}.  The field is defined as  
$\Phi_m(f):= \Phi_S(E_m(\Re(f))) +i\Phi_S(E_m(\Im(f)))$  for any  test function$f\in\cS(\mathbb{R}^4)$, 
where $\Phi_S$ is the Segal quantization of $L^2(\rH_m,\rd\Omega_m)$  and
$E_m:\cS(\mathbb{R}^4)\to L^2(\rH_m,\rd\Omega_m)$ is defined as 
\begin{equation}
\label{Ceb:8}
E_m(f):= (2\pi)^{-2}\int e^{i\, p\cdot x} f(x) \rd^4x \restriction \rH_m \ ,  
\end{equation}
where $p\cdot x= p_0t - \mathbf{p}\cdot\mathbf{x}$. So $E_m(f)$ is nothing but the restriction  to the hyperboloid of mass $m$ of the  Fourier transformation of $f$.  \\
\indent  We now are ready to prove the existence of non-trivial,  causal and continuously 
$\mathscr{P}^\uparrow_+$-covariant connections. 
\begin{theorem}
\label{Ceb:9}
Let $K$ be the set of double cones in the Minkowski 
spacetime. Then:  
\begin{itemize}
\item[(i)]  there exists a non-flat,   causal and continuously $\mathscr{P}^\uparrow_+$-covariant connection of $K$; 
\item[(ii)]  the net of causal loops $(\cA,\mathrm{id},\alpha)_K$ has a non-trivial and  continuously 
$\mathscr{P}^\uparrow_+$-covariant representation.  
\end{itemize}
\end{theorem}
\begin{proof}
$(ii)$ follows from  $(i)$  and Lemma \ref{Cb:2}. So let us prove $(i)$.  Consider the free Hermitian scalar field 
$(\Phi_m,\Gamma_m,\cH_m)$. Let $f$ be the invariant 
$0$-cochain constructed in Proposition \ref{Cea:14}. As observed in that proposition 
the twisted cobord $\delta f$ is an element of \smash{$\rC^1(K,\cD)^{\mathscr{P}^\uparrow_+}$}, 
$\delta f^{odd}_b=f_{\partial_0b}-f_{\partial_1b}$ is different from zero on 1-simplices 
having different faces, and $\delta f^{ev}_b=f_{|b|}$ is different from zero for any 1-simplex. 
Using the formula (\ref{Ceb:7}), define   
\[
 [\Phi_m \delta f](b):= 
 \exp\left\{i \int_{b_\circledcirc }|(\delta f)^{ev}_b(x)| \rd^4x \cdot \Phi_m\left((\delta f)^{odd}_b \right)\right\}  
 \ , \qquad 
 b\in\Si_1(K) \ . 
\]
By Theorem \ref{Ceb:6}, the pair $([\Phi_m \delta f],\Gamma_m)$ is a 
causal and continuously $\mathscr{P}^\uparrow_+$-covariant connection.  
We now want to show that  this connection is not trivial. This amounts to showing that 
$\Phi_m((\delta f)^{odd}_b)\ne 0$ or, equivalently, that $E_m((\delta f^{odd})_b)\ne 0$ 
for some 1-simplex $b$ which is does not belong to the nerve of $K$.  \\
\indent We start by showing that $E_m(f_a)\ne 0$ 
for any 0-simplex $a$. Afterwards we shall prove that there are 1-simplices $b$ such that 
$E_m((\delta f^{odd})_b)\ne 0$. Recall that  the 0-cochain $f$ is obtained 
by letting act  the Poincar\'e group on the set of functions $f_R$, $R>0$ which are associated to double cones 
whose base lays on the subspace at $t=0$  of the reference frame of the Minkowski spacetime. Recalling the definition of 
these  functions (\ref{Cea:14a}) we observe that  
\[
\hat f_R(p):=\int e^{i\, p\cdot x} f_R(x) \rd^4x  = f'_R(p_0) \cdot  f''_R(\mathbf{p}) \ , 
\]
where 
\[
  f'_R(p_0):=  \int e^{ip_0t} h(2t/R) \rd t\ \  , \ \   f''_R({\bf p}):=  \int e^{-i {\bf p\cdot x}} h(4 |{\bf x}|^2/R^2) \rd^3 x \ . 
\]
Let us study  the zeroes of the function $\hat f_R$. Note that  $p$ is a zero of  $\hat f_R$  if, and only if,  either $p_0$ is a zero of 
$f'_R$ or $\textbf{p}$ is a zero of $f''_R$. 
Observe that both $f'_R$ and $f''_R$  are holomorphic entire functions, since $h$ is a compactly supported smooth function. 
Therefore  the zeroes of $f'$ are at most a countable set 
of $\mathbb{R}$.  Instead,  since $f''_R$ is, up to a factor,  the Fourier transform  of a spherically symmetric function,   
if $f''_R({\bf q})=0$ then $f''_R({\bf p})=0$ for any ${\bf p}\in\mathbb{R}^3$ such that $|{\bf p}|=|{\bf q}|= r$. 
So the zeroes are distributed on spherical surfaces and the collection $Z$ of the radius $r>0$ of such surfaces,    
cannot have an accumulation point. 
In fact, if it were so,  being  the function  $z\mapsto f''_R(z,0,0)$  holomorphic entire, since 
$f''_R(r,0,0)=0$ for any $r\in Z$, we should have that $f''_R(z,0,0)=0$ for any  $z\in\mathbb{C}$. 
This implies that 
any $r\geq 0$ belongs to $Z$, so $f''_R=0$ and this leads to a contradiction because 
$f''_R$ is the Fourier transform of a non-zero smooth function. In conclusion we have that 
$E_m(f_R)\ne 0$ for any $R>0$, because $\hat f_R=0$, at most  on a countable  collection  of measure zero subsets of  $\rH_m$, associated to the  
zeroes of $f'_R$ and of $f''_R$. This also implies that $E_m(f_a)\ne 0$ in $L^2(\rH_m,\rd\Omega_m)$ since 
the 0-cochain $f$ is obtained by letting act the Poincar\'e group of the function $f_R$ for any $R>0$.\\ 
\indent Consider now the following 1-simplex $b$: $\partial_1b:=o_R$;  $\partial_0b:=o_R+y$ with $0\ne y\in\mathbb{R}^4$;
$|b|$ is a  double cone greater than $\partial_1b\cup\partial_0b$ .
Since $f_{\partial_0b}(x)= f_R(x-y)$, we have that  
\[
\hat f_{\partial_0b} (p)=\int e^{i\, p\cdot x} f_R(x-y) \rd^4x  = e^{ip\cdot y} \hat f_R(p)=  e^{ip\cdot y} \hat f_{\partial_1b}(p) \ . 
\]
Hence 
\[
E_m((\delta f)^{odd}_b) =  E_m(f_{\partial_0b}- f_{\partial_1b})= E_m(f_{\partial_1b})\cdot (e^{ip\cdot y}-1)\restriction \rH_m  
\]
and this is not zero in $L^2(\rH_m,\rd\Omega_m)$.  This shows that the connection 1-cochain 
$ [\Phi_m \delta f]$ is not trivial. In addition  
$ [\Phi_m \delta f]$  results to be not flat, i.e.\ does not satisfy the 1-cocycle identity on all $\Si_1(K)$. In fact  
$ [\Phi_m \delta f](b)$  depends in general  on the support $|b|$ of the 1-simplex $b$ through the factor 
$\int_{b_\circledcirc }|(\delta f)^{ev}_b| \rd\mu$, according to the definition of $b_\circledcirc$.
\end{proof}
\begin{remark}
\label{Ceb:10}
The results of this section cannot be fully applied to the case of the circle $S^1$ 
with respect to M\"obius group. 
In fact, even if Proposition \ref{Cea:15} shows that the space \smash{$\rC^1(K,\cD)^{\mbox{M\"ob}}$} is not zero, 
there does not exist a M\"obius-invariant measure on $S^1$, so the formula 
(\ref{Ceb:7}) cannot be applied.  
There are two ways to avoid this problem.
The first one is to consider the rotation subgroup of M\"ob  as a symmetry group,
which clearly admits an invariant measure.
The second approach consists in modifying (\ref{Ceb:7}):
as a first step, we consider \smash{$f \in \rC^1(K,\cD)^{\mbox{M\"ob}}$} and
make the cutoff
\[
\tilde{f}_b := 
\left\{
\begin{array}{ll}
0 \ \ , \ \ b \in \rN_1(K)
\\
f_b \ \ , \ \ b \in \rT_1(K)
\ .
\end{array}
\right.
\]
The above definition is well-posed since $\rN_1(K)$ and $\rT_1(K)$ are globally stable
under the ${\mbox{M\"ob}}$-action and have void intersection, so there is no risk to
define $\tilde{f}$ in different ways over some 1-simplex.
This allows us to define
\[
[\Phi f]_{odd}(b)
\ := \ 
\exp\left\{i \Phi \left( \tilde{f}^{odd}_b \right) \right\}  
\ , \qquad 
b\in\Si_1(K) \ .
\]
The above expression yields a causal and covariant connection which is expected to be 
non-trivial when $f^{odd}$ is nonzero on $\rT_1(K)$. The price to pay for the construction
of $[\Phi f]_{odd}$ is the loss of any explicit information on the support (corona) of $b$.
In conclusion, non-triviality of the invariant subspace of $\rC^1(K,\cD)$ under the symmetry group action 
seems to be the only strong hypothesis in Theorem \ref{Ceb:6}, under which
we expect that a version of the  result holds in any spacetime where a causal, selfadjoint 
Wightman field can be defined.
\end{remark}

\section{Concluding remarks}

In this paper we gave a model independent construction  of a causal  and covariant net of $\rC^*$-algebras over a spacetime,
called the net of causal loops. 
This is constructed using as input  only the underlying spacetime; more precisely, 
a base $K$ of the topology of the spacetime encoding the causal and the symmetry structure. 
Generators of the local algebras of the net are free groups of loops of the  poset  $K$ (the base $K$ ordered under  inclusion). \\
%
%
\indent We showed that representations of nets of causal loops appear associated with quantum fields, 
under limitations that seems to be imposed only by the geometry of the spacetime and the symmetry group.
Two remarkable examples have been constructed: the case of Minkowski spacetime and the case of $S^1$. 
In the Minkowski spacetime, this also has led to the prove of the existence of covariant representations 
satisfying the spectrum condition. Other insights could come from the theory of subsystems \cite{CDR,CC1,CC2} 
since the net of causal loops defined in such representations  turns out to be a subnet of the net generated by the
given quantum field. 
Furthermore, having in mind applications to gauge theories,  refining the construction associating quantum fields to representations of the net of causal loops, we hope to produce more interesting representations at the physical level. In particular, the notion of connection system may be
a good starting point for a definition of local gauge transformation in the setting of algebraic quantum field theory.\\
\indent At the mathematical level, a related question is to classify the above mentioned representations in terms
of the equivariant cohomology of the complex defined by the nets of test functions $C^*(K,\cD)^S$, $*=1,2,\ldots$  
and, in particular, to give conditions for irreducibility in terms of properties of the translator.
This could give some light on the ``charge" structure exhibited by nets of causal loops.\\ 
\indent  A further  interesting direction of research is, in the case the spacetime has a topological symmetry group, 
the definition of a topology on the free group of loops making continuous the action of the symmetry.  
This could allow a sort of \emph{Mackey's  induction}  in the sense that,  for a given representation of the symmetry  
group,   a continuously covariant representation of the group of loops (i.e.\  a continuously covariant representation of the net of causal loops) should be induced.  
Up to now, the authors obtained partial results on this topic, 
that presents difficulties arising  from its combinatorial free group aspect.

\
\\[5pt]
\noindent {\small \textbf{Acknowledgements.} The authors would like to thank F. Fidaleo, D. Guido and G. Morsella for 
 many fruitful discussions on the topics treated in this paper.}


\begin{thebibliography}{99}
\markboth{Bibliography}{Bibliography}


\bibitem{Ara}
H. Araki. {\em Mathematical theory of quantum fields.}  Oxford University Press, Oxford, 2009.







\bibitem{BrF}
R. Brunetti, K. Fredenhagen:
Algebraic approach to Quantum Field Theory.
To appear on Elsevier Encyclopedia of Mathematical Physics.
Available as  \texttt{arXiv:math-ph/0411072v1}. 




\bibitem{BF} 
 D. Buchholz, K. Fredenhagen: Locality and the structure of particle states. 
 Commun. Math. Phys. {\bf 84}, (1982), 1--54.






\bibitem{BFV}
R. Brunetti, K. Fredenhagen, R. Verch:
The generally covariant locality principle -- A new paradigm
for local quantum physics. Commun. Math. Phys. {\bf 237}, (2003), 31--68. 
Available as  \texttt{arXiv:math-ph/0112041v1}.


\bibitem{BGL}
R. Brunetti, D. Guido, R. Longo: 
Modular localization and wigner particles.
Rev.Math.Phys. {\bf 14} (2002) 759--785. Available as 
\texttt{arXiv:math-ph/0203021v2}.




\bibitem{BH}
D. Buchholz, R. Haag: The Quest for Understanding in Relativistic Quantum Physics 
J. Math. Phys. {\bf 41} (2000), no. 6, 3674--3697. 
Available as \texttt{arXiv:hep-th/9910243v2}.







\bibitem{BLS}
D. Buchholz, G. Lechner, S. J. Summers:  Warped convolutions, Rieffel deformations and the construction of quantum field theories. Comm. Math. Phys. {\bf 304} (2011), no. 1, 95–-123. Available as 	\texttt{arXiv:1005.2656v1 [math-ph]}. 


\bibitem{BLT}
N. N. Bogolubov, A. A. Logunov, I. T. Todorov. {\em Introduction to axiomatic quantum field theory.} 
Translated from the Russian by Stephen A. Fulling and Ludmila G. Popova. Edited by Stephen A. Fulling. Mathematical Physics Monograph Series, No. 18. W. A. Benjamin, Inc., Reading, Mass.-London-Amsterdam, 1975.


\bibitem{Bor}
H. J. Borchers: On Structure of the Algebra of Field Operators, Nuovo Cimento, {\bf 24}, 1418–-1440 (1962)

\bibitem{Bos}
H. Bostelmann: Phase space properties and the short distance structure in quantum field theory. 
J. Math. Phys. {\bf 46} (2005), no. 5, 052301, 17 pp. Available as \texttt{arXiv:math-ph/0409070v3}.


\bibitem{BR} 
R. Brunetti, G. Ruzzi: 
Quantum charges and spacetime topology: The emergence of new superselection sectors. 
Comm. Math. Phys. {\bf 287} (2009), no. 2, 523--563. 
Available as  \texttt{arXiv:0801.3365 [math-ph]} 

\bibitem{Buc}
D. Buchholz:
On quantum fields that generate local algebras. 
J. Math. Phys. {\bf 31} (1990), no. 8, 1839–-1846. 


%

\bibitem{CC1}
S. Carpi, R. Conti: 
Classification of Subsystems for Local Nets with Trivial Superselection Structure.
Commun. Math. Phys.,  {\bf 217},  (2001),  89--106.
Available as \texttt{arXiv:math/0002204v1 [math.OA] }

\bibitem{CC2}
S. Carpi, R. Conti: 
Classification of subsystems for graded-local nets with trivial superselection structure 
Commun. Math. Phys.,  {\bf 253},  (2005),  423--449.
Available as \texttt{arXiv:math/0312033v1 [math.OA]}


\bibitem{CDR}
R. Conti, S. Doplicher, J. E. Roberts: 
Superselection Theory for Subsystems
Commun. Math. Phys.,  {\bf 218},  (2001),  263--281. 
Available as \texttt{arXiv:math/0001139v1 [math.OA]}.




\bibitem{DHR}
  S. Doplicher, R. Haag,  J. E. Roberts.
     {\em Local observables and particle statistics I.}
   Commun. Math Phys. {\bf 23}, (1971), 199--230.
   {\em Local observables and particle statistics II.}
    Commun. Math Phys. {\bf 35}, (1974), 49--85 .

%
%
%
%


\bibitem{DopConf}
S. Doplicher:  "Superselection structure in Local Quantum Theories with (neutral) massless particle" 
Talk given at conference "Modern Trends in AQFT", Pavia, Italy, September 2011. 

\bibitem{DR}
S. Doplicher,  J. E. Roberts : Why there is a field algebra with a compact gauge group describing the superselection structure in particle physics. Comm. Math. Phys. {\bf 131} (1990), no. 1, 51-–107.




\bibitem{Few}
C. J. Fewster: On the notion of "the same physics in all spacetimes".
To appear in the Proceedings of the conference "Quantum field theory and gravity", Regensburg (28 Sept - 1 Oct 2010). 
Available as \texttt{arXiv:1105.6202v2 [math-ph]}. 


\bibitem{FV}
C.  J.  Fewster, R. Verch:
Dynamical locality and covariance: What makes a physical theory the same in all spacetimes?
Preprint \texttt{	arXiv:1106.4785v1 [math-ph]}.



\bibitem{GJ}
J. Glimm, A. Jaffe. {\em Quantum physics. A functional integral point of view.}
Second edition. Springer-Verlag, New York, 1987.


\bibitem{GL}
H. Grosse, G. Lechner: Noncommutative deformations of Wightman quantum field theories. 
J. High Energy Phys. (2008), no. 9, 131, 29 pp. 
Available as \texttt{arXiv:0808.3459v1 [math-ph]}. 


\bibitem{GLRV}
D. Guido, R. Longo, J. E. Roberts, R. Verch:
Charged sectors, spin and statistics in quantum field theory on curved spacetimes.  
Rev. Math. Phys. {\bf 13} (2001), no. 2, 125--198. 
Available as  \texttt{arXiv:math-ph:9906019}.


\bibitem{Haa}
R. Haag. {\em Local Quantum Physics.}
Springer Texts and Monographs in Physics,  1996,  2nd edition.



\bibitem{Joh}
D. L. Johnson.  {\em Presentations of groups.}  Second edition. London Mathematical Society Student Texts, 15. 
Cambridge University Press, Cambridge, 1997.

\bibitem{K}
Y. Kawahigashi: 
From Operator Algebras to Superconformal Field Theory.
 J. Math. Phys. {\bf 51}  (2010), no. 1, 015209, 20 pp.. 
Available as \texttt{arXiv:1003.2925v1 [math-ph]}.






\bibitem{Pi}
N. Pinamonti:
Conformal generally covariant quantum field theory: the scalar field and its Wick products. 
Comm. Math. Phys. {\bf 288} (2009), no. 3, 1117-1135. Available as \texttt{	arXiv:0806.0803v2 [math-ph]}.

\bibitem{Rob0}
J. E. Roberts:  A survey of local cohomology. {\em Mathematical problems in theoretical physics} 
(Rome, 1977),  81–-93, 
Lecture Notes in Phys., 80, Springer, Berlin-New York, 1978. 


\bibitem{Rob}
J. E. Roberts: Lectures on algebraic quantum field theory. {\em The algebraic theory of superselection sectors} (Palermo, 1989), 1--112, World Sci. Publ., River Edge, NJ, 1990. 


\bibitem{Rob2}
J. E. Roberts:
More lectures in algebraic quantum field theory
In: S. Doplicher, R. Longo (eds.)  {\em Noncommutative geometry}
C.I.M.E. Lectures, Martina Franca, Italy, 2000.
Spinger 2003.



\bibitem{RR}
   J. E. Roberts, G. Ruzzi:
   A cohomological description of connections and curvature over posets.
    Theo. App. Cat.  \textbf{16}  (2006), No.30, 855--895.
     Available as \texttt{	arXiv:math/0604173v1 [math.AT]}.

\bibitem{RRV}  
 J. E. Roberts, G. Ruzzi, E. Vasselli: A theory of bundles over posets. 
 Adv. Math. {\bf 220} (2009), no. 1, 125--153. 
 Available as \texttt{arXiv:0707.0240 [math.AT]}.

\bibitem{RS}
M. Reed, B. Simon.  {\em Methods of modern mathematical physics. II. Fourier analysis, self-adjointness.} 
Academic Press, New York-London, 1975

\bibitem{Ruz}
G. Ruzzi: Homotopy of posets, net-cohomology, and theory of superselection sectors in globally hyperbolic  spacetimes.  Rev. Math. Phys. {\bf 17} (9), 1021--1070 (2005). 
Available as \texttt{arXiv:math-ph/0412014}. 

\bibitem{RV}
G. Ruzzi, E. Vasselli: A new light on nets of C*-algebras and their representations.
Preprint \texttt{arXiv:1005.3178v3 [math.OA]}.



\bibitem{Sch}
B. Schroer: Modular localization and the bootstrap–formfactor program, Nucl. Phys. B, {\bf 499}, 547–-568 (1997).

\bibitem{Sor}
R. D. Sorkin:  Causal sets: discrete gravity. 
{\em Lectures on quantum gravity}, 305--327, Ser. Cent. Estud. Cient., Springer, New York, 2005.

\bibitem{Sur}
S. Surya: Directions in Causal Set Quantum Gravity.  
 Preprint  \texttt{arXiv:1103.6272v1  [gr-qg]}.  


\bibitem{SW}
R. F. Streater,  A. S. Wightman.  {\em PCT, spin and statistics, and all that.} 
Corrected third printing of the 1978 edition. Princeton Landmarks in Physics. Princeton University Press, Princeton, NJ, 2000. 



\bibitem{Uhl}
A. Uhlmann: \"Uber die Definition der Quantenfelder nach Wight- man und Haag.  Wiss. Zeitschr. Karl-Marx Univ. Leipzig, {\bf 11}, 213–-217 (1962)

\bibitem{Wal}
R. M. Wald. {\em Quantum field theory in curved spacetime and black hole thermodynamics.} 
Chicago Lectures in Physics. University of Chicago Press, Chicago, IL, 1994. 



\end{thebibliography}
\end{document}